\documentclass[11pt]{article}

\usepackage{naysh}

\usepackage{mdframed}
\usepackage{graphicx}
\usepackage{float,wrapfig}
\usepackage{subcaption}

\begin{document}
\title{A Local-to-Global Theorem for Congested Shortest Paths}

\author{
Shyan Akmal\thanks{MIT CSAIL, \url{naysh@mit.edu}. Supported in part by NSF grant CCF-2129139. }
\and 
Nicole Wein\thanks{DIMACS, Rutgers University, \url{nicole.wein@rutgers.edu}.
Supported by a grant to DIMACS from the Simons Foundation (820931).} 
}

\date{\vspace{-3ex}}

\maketitle

\thispagestyle{empty}

\begin{abstract}
    Amiri and Wargalla proved the following local-to-global theorem about shortest paths in directed acyclic graphs (DAGs): if $G$ is a weighted DAG with the property that for each subset $S$ of 3 nodes there is a shortest path containing every node in $S$, then there exists a pair $(s,t)$ of nodes such that there is a shortest $st$-path containing every node in $G$. 
   We extend this theorem to general graphs. 
   For undirected graphs, we prove that the same theorem holds (up to a difference in the constant 3). 
   For directed graphs, we provide a counterexample to the theorem (for any constant). 
   However, we prove a \emph{roundtrip} analogue of the theorem which guarantees there exists a pair $(s,t)$ of nodes such that every node in $G$ is contained in the union of a shortest $st$-path and a shortest $ts$-path. 
   
   The original local-to-global theorem for DAGs has an application to the \textsf{$k$-Shortest Paths with Congestion $c$} (\textsf{($k,c$)-SPC}) problem.
   In this problem, we are given a weighted graph $G$, together with $k$ node pairs $(s_1,t_1),\dots,(s_k,t_k)$, and a positive integer $c\leq k$, and tasked with finding a collection of paths $P_1,\dots, P_k$ such that each $P_i$ is a shortest path from $s_i$ to $t_i$, and every node in the graph is on at most $c$ paths $P_i$, or reporting that no such collection of paths exists. 
   When $c=k$, there are no congestion constraints, and the problem can be solved easily by running a shortest path algorithm for each pair $(s_i,t_i)$ independently. 
   At the other extreme, when $c=1$, the \textsf{$(k,c)$-SPC} problem is equivalent to the \textsf{$k$-Disjoint Shortest Paths} (\textsf{$k$-DSP}) problem, where the collection of shortest paths must be node-disjoint. 
   For fixed $k$, \textsf{$k$-DSP} is polynomial-time solvable on DAGs and undirected graphs.  
   Amiri and Wargalla interpolated between these two extreme values of $c$, to obtain an algorithm for \textsf{$(k,c)$-SPC} on DAGs that runs in polynomial time when $k-c$ is constant.
   
   In the same way, we prove that \textsf{$(k,c)$-SPC} can be solved in polynomial time on undirected graphs whenever $k-c$ is constant.
For directed graphs, because of our counterexample to the original theorem statement, our roundtrip local-to-global result does not imply such an algorithm \textsf{$(k,c)$-SPC}. 
   Even without an algorithmic application, our proof for directed graphs may be of broader interest because it characterizes  intriguing structural properties of shortest paths in directed graphs. 
   
   
\end{abstract}
\clearpage
\pagenumbering{arabic}


\section{Introduction}

An intriguing question in graph theory and algorithms is: ``can we understand the structure of shortest paths in (directed and undirected) graphs?'' 
Or more specifically: ``can we understand the structure of the \emph{interactions} between shortest paths in graphs?'' 
This question has been approached from various angles in the literature. 
For instance, Bodwin \cite{DBLP:conf/soda/Bodwin19} characterizes which sets of nodes can be realized as \emph{unique} shortest paths in weighted 
graphs, and Cizma and Linial investigate the properties of graphs whose shortest paths satisfy or violate certain geometric properties \cite{Cizma2022,Cizma2023}. 
Additionally, there is a large body of work on distance preservers, where the goal is to construct a subgraph that preserves distances in the original graph (see \cite{DBLP:conf/soda/Bodwin17} and the references therein). 
There are also numerous computational tasks in which structural results about shortest paths inform the creation and analysis of algorithms, including distance oracle construction \cite{DBLP:conf/soda/ChechikZ22}, the $k$-disjoint shortest paths problem, \cite{SDP-undirected-geometric}, and the next-to-shortest path problem \cite{DBLP:conf/cocoa/Wu10}
(the cited papers are the most recent publications in their respective topics).

The angle of our work 
is inspired by the following local-to-global theorem\footnote{This is slightly different from the statement of the theorem in \cite{SDP-congestion-DAG}, where the authors write present their result in the context of the \textsf{Shortest Paths with Congestion} problem, which we discuss in detail later.} of Amiri and Wargalla \cite{SDP-congestion-DAG} concerning the structure of shortest paths in directed acyclic graphs (DAGs):

\begin{theorem}[\cite{SDP-congestion-DAG}]\label{thm:dagnew}
Let $G$ be a weighted DAG with the property that, for each set $S$ of 3 nodes, there is a shortest path containing every node in $S$. 
Then, there exists a pair $(s,t)$ of nodes such that there is a shortest $st$-path containing every node in $G$.
\end{theorem} 

Note that in the statement of \Cref{thm:dagnew}, the condition that ``for each set $S$ of 3 nodes, there is a shortest path containing every node in $S$'' is equivalent to the condition that one of the 3 nodes is on a shortest path between the other two.

\cref{thm:dagnew} is ``local-to-global'' in the sense that from a highly congested local structure (every small subset of nodes is contained in a shortest path) we deduce a global structure (all nodes in the graph live on a single shortest path).

At first glance, \cref{thm:dagnew} may appear rather specialized, since the existence of shortest paths through \emph{all} triples of nodes is a rather strong condition. 
However, Amiri and Wargalla \cite{SDP-congestion-DAG} show that \cref{thm:dagnew} has applications to the \textsf{$k$-Shortest Paths with Congestion $c$} (\textsf{($k,c$)-SPC}) on problem on DAGs: 
in this problem we are given a DAG, and are tasked with finding a collection of shortest paths between $k$ given source/sink pairs, such that each node in the graph is on at most $c$ of the paths. 
We discuss this application in detail in \Cref{sec:back}.

Amiri and Wargalla \cite{SDP-congestion-DAG} raised the question of whether their results can be extended to general graphs, both undirected and directed. 
Our work answers this question.

\subsection{Structural Results}

We ask the following question:

\begin{center}\emph{Is \cref{thm:dagnew} true for general (undirected and directed) graphs?}\end{center}

\noindent Our first result  answers this question  affirmatively for undirected graphs (with constant 4 instead of 3). 

\begin{theorem}[Undirected graphs]\label{thm:undirnew}
Let $G$ be a weighted undirected graph with the property that, for each set $S$ of 4 nodes, there is a shortest path containing every node in $S$. 
Then, there exists a pair $(s,t)$ of nodes  such that some shortest $st$-path contains every node in $G$.
\end{theorem}

The constant $4$ in the statement of \Cref{thm:undirnew} cannot be replaced with $3$, as seen by considering an undirected cycle on four vertices.

\cref{thm:undirnew} implies a faster algorithm for \textsf{$(k,c)$-SPC} on \emph{undirected} graphs, answering an open question raised in \cite[Section 4]{SDP-congestion-DAG}.
We discuss the background of this problem and our improvement in \cref{sec:back}.

Our second result is for directed graphs. 
First, we observe that there is actually a \emph{counterexample} to \cref{thm:dagnew} for general directed graphs: 
let $a$ be the constant for which we desire a counterexample (i.e., the constant that is 3 in \cref{thm:dagnew}, and 4 in \cref{thm:undirnew}). 
Let $G$ be a cycle with bidirectional edges, where all clockwise-pointing edges have weight 1 and all counterclockwise-pointing edges have weight $a$. 
One can verify that the precondition of the theorem holds: for each set $S$ of $a$ nodes there is a shortest path containing every node in $S$, just by taking the shortest clockwise path through all nodes in $S$. 
However, no single shortest path contains every node in $G$, so the theorem does not hold.

Even though a direct attempt at generalizing \cref{thm:dagnew} to all directed graphs fails, one might hope for some analogue of \cref{thm:dagnew} that does hold.
One interpretation of the above counterexample is that the exact statement of \cref{thm:dagnew} is not the ``right'' framework for getting a local-to-global shortest path phenomenon in general directed graphs. 
To that end, we consider the \emph{roundtrip} analogue of \cref{thm:dagnew}, where the final path through every node is a shortest roundtrip path, i.e., the union of a shortest $st$-path and a shortest $ts$-path 
(roundtrip distances are a common object of study in directed graphs, with there being much research, for example, in roundtrip routing \cite{DBLP:conf/soda/CowenW99}, roundtrip spanners \cite{DBLP:journals/talg/RodittyTZ08}, and roundtrip diameter computation \cite{DBLP:conf/soda/AbboudWW16}). 
Note that the above counterexample does not apply to the roundtrip analogue of \cref{thm:dagnew} since there exists a pair $(s,t)$ of nodes such that the union of a shortest $st$-path and a shortest $ts$-path are both in the clockwise direction and thus contain all nodes in the graph. 

For our second result, we present a roundtrip analogue of \cref{thm:dagnew} which holds true for general directed graphs (with the constant 11 instead of 3):

\begin{theorem}[Directed graphs]
\label{thm:dirnew}
Let $G$ be a weighted directed graph with the property that, for each set $S$ of 11 nodes, there is a shortest path containing every node in $S$. Then, there exists a pair $(s,t)$ of nodes such that the union of a shortest $st$-path and a shortest $ts$-path contains every node in $G$.
\end{theorem}

Proving \cref{thm:dirnew} requires overcoming a number of technical challenges involving the c omplex structure of shortest paths in directed graphs.
Due to its roundtrip nature, unlike \Cref{thm:undirnew}, \Cref{thm:dirnew} does not appear to have any immediate algorithmic applications.

\subsection{Disjoint and Congested Shortest Path Problems}\label{sec:back}

In this section we introduce the \textsf{$k$-Shortest Paths with Congestion $c$} (\textsf{($k,c$)-SPC}) problem
and state the implications of our work for this problem. 

\subsubsection{Background}
We begin by discussing the related \textsf{$k$-Disjoint Shortest Paths ($k$-DSP)} problem. 
For more related work on disjoint path problems in general, see \cref{app:rel}.

\paragraph*{Disjoint Shortest Paths}

Formally, the \textsf{$k$-Disjoint Shortest Paths ($k$-DSP)} problem is defined as follows:

\begin{mdframed}{\textsf{$k$-Disjoint Shortest Paths ($k$-DSP)}: Given a graph $G$ and $k$ node pairs $(s_1,t_1),\dots,(s_k,t_k)$, find a collection of node-disjoint paths $P_1,\dots, P_k$ such that each $P_i$ is a shortest path from $s_i$ to $t_i$, or report that no such collection of paths exists.}\end{mdframed}

The \textsf{$k$-DSP} problem was introduced in the 90s by Eilam-Tzoreff \cite{eilam1998disjoint}, who gave a polynomial-time algorithm for undirected graphs when $k=2$, and conjectured that there is a polynomial-time algorithm for any fixed $k$ in both undirected and directed graphs. 
Recently, Lochet \cite{SDP-undirected-original} proved Eilam-Tzoreff's conjecture for undirected graphs by showing that \textsf{$k$-DSP} can be solved in polynomial time for any fixed $k$. Subsequently, the dependence on $k$ in the running time was improved by Bentert, Nichterlein, Renken, and Zschoche \cite{SDP-undirected-geometric}. 
It is known that \textsf{$k$-DSP} on undirected graphs is \textsf{W[1]}-hard  \cite[Proposition 36]{SDP-undirected-geometric},
so this problem is unlikely to be fixed-parameter tractable. 

For directed graphs, B{\'{e}}rczi and Kobayashi \cite{DBLP:conf/esa/Berczi017} showed that \textsf{$2$-DSP} can be solved in polynomial time.
For $k\ge 3$ however, determining the complexity of \textsf{$k$-DSP} on directed graphs remains a major open problem. 
This problem is only known to be polynomial-time solvable for special classes of directed graphs, such as DAGs and planar graphs \cite{DBLP:conf/esa/Berczi017}.

\paragraph*{Shortest Paths with Congestion}

The \textsf{$k$-Shortest Paths with Congestion $c$} (\textsf{($k,c$)-SPC}) problem is the variant of \textsf{$k$-DSP} where some amount of \emph{congestion} (paths overlapping at nodes) is allowed. 
In general, problems of finding paths with limited congestion are well-studied studied in both theory and practice. 
For instance, there is much work on the problem in undirected graphs where the goal is to find a maximum cardinality subset of node pairs $(s_i,t_i)$ that admit (not necessarily shortest) paths with congestion at most $c$ \cite{DBLP:journals/combinatorica/RaghavanT87, DBLP:conf/ipco/KolliopoulosS98, DBLP:journals/mor/BavejaS00, DBLP:journals/algorithmica/AzarR06, DBLP:journals/siamcomp/AndrewsZ07, DBLP:journals/siamcomp/ChekuriKS09, DBLP:conf/focs/Andrews10, DBLP:conf/stoc/KawarabayashiK11, DBLP:conf/soda/ChekuriE13,DBLP:journals/siamcomp/Chuzhoy16}. 
As another example, \cite{kawarabayashi2014excluded} provides, for fixed $k$, a polynomial-time algorithm for the problem on directed graphs of determining that either there is no set of disjoint paths between the node pairs $(s_i,t_i)$, or finding a set of such paths with congestion at most 4. 
Another example for directed graphs is the problem of finding paths between the node pairs $(s_i,t_i)$ where only some nodes in the graph have a congestion constraint \cite{DBLP:journals/tcs/LopesS22}. 

Formally, the \textsf{$k$-Shortest Paths with Congestion $c$ ($(k,c)$-SPC)} problem is defined as follows:

\begin{mdframed}{\textsf{$k$-Shortest Paths with Congestion $c$ (($k,c$)-SPC)}: Given a graph $G$, along with $k$ node pairs \\$(s_1,t_1),\dots,(s_k,t_k)$, and a positive integer $c\leq k$, find a collection of paths $P_1,\dots, P_k$ such that each $P_i$ is a shortest path from $s_i$ to $t_i$, and every node in $V(G)$ is on at most $c$ paths $P_i$, or report that no such collection of paths exists.}\end{mdframed}

The \textsf{$(k,c)$-SPC} problem was introduced by Amiri and Wargalla \cite{SDP-congestion-DAG}. 
Before that, the version of \textsf{$(k,c)$-SPC} where the paths are not restricted to be shortest paths was studied by Amiri, Kreutzer, Marx, and Rabinovich \cite{DBLP:journals/ipl/AmiriKMR19}.

When $c=1$, the \textsf{$(k,c)$-SPC} problem is equivalent to the \textsf{$k$-DSP} problem. 
At the other extreme, when $c=k$, there are no congestion constraints, so the problem can be easily solved in polynomial time by simply finding a shortest path for each pair $(s_i,t_i)$ independently. Amiri and Wargalla \cite{SDP-congestion-DAG} asked the following question: \emph{can we interpolate between these two extremes?} 
In particular, can we get algorithms for \textsf{$(k,c)$-SPC} where the exponential dependence on $k$ for \textsf{$k$-DSP} can be replaced with some dependence on $O(k-c)$ instead?

Amiri and Wargalla \cite{SDP-congestion-DAG} achieved this goal for DAGs. In particular, they gave a reduction from \textsf{$(k,c)$-SPC} on DAGs to \textsf{$k$-DSP} on DAGs of the following form: 
letting $d=k-c$, if \textsf{$k$-DSP} on DAGs can be solved in time $f(n,k)$, then \textsf{$(k,c)$-SPC} on DAGs can be solved in time $O\left(\binom{k}{3d}\cdot f(2dn, 3d)\right)$.
The essential aspect of this running time is that the second input to the function $f$ is not $k$ but rather an $O(d)$ term.
A key tool in their reduction is \cref{thm:dagnew} (stated in a different way). 

Since \textsf{$k$-DSP} can be solved in $n^{O(k)}$ time on DAGs \cite{DBLP:conf/esa/Berczi017}, the above result implies that \textsf{$(k,c)$-SPC} can be solved in $\binom{k}{3d}\cdot (2dn)^{O(d)}$ time on DAGs. 
That is, \textsf{$(k,c)$-SPC} on DAGs is polynomial-time solvable for arbitrary $k$ whenever $d$ is constant.
We note that for for every $c$, the \textsf{$(k,c)$-SPC} problem on DAGs is \textsf{W[1]}-hard with respect to $d$, so the problem is unlikely to be fixed-parameter tractable with respect to $d$ \cite[Proof of Theorem 3]{SDP-congestion-DAG}.

\subsubsection{Algorithmic Results}

Similarly to Amiri and Wargalla's result for DAGs, our result for undirected graphs, \cref{thm:undirnew}, implies a reduction from \textsf{$(k,c)$-SPC} to \textsf{$k$-DSP} on undirected graphs.

\begin{restatable}{lemma}{thmundir}\label{thm:undir}
    If \textsf{$k$-DSP} can be solved in $f(n,k)$ time on undirected graphs, then \textsf{$(k,c)$-SPC} can be solved in $O\left(\binom{k}{4d}\cdot f(3dn, 4d)\right)$ time on undirected graphs.
\end{restatable}

\cref{thm:undir} follows from \cref{thm:undirnew} using an argument nearly identical to the one presented for DAGs in \cite{SDP-congestion-DAG} (up to a difference in constants). 
For completeness, we include a full proof of this result in \cref{app:cor}.

Since it is known that \textsf{$k$-DSP} can be solved in undirected graphs in time $n^{O(k\cdot k!)}$ \cite{SDP-undirected-geometric}, applying \Cref{thm:undirnew} together with \Cref{thm:undir}, we deduce the following result.

\begin{corollary}\label{cor}
    \textsf{$(k,c)$-SPC} on undirected graphs can be solved in
    $\binom{k}{4d}\cdot (3dn)^{O(d\cdot (4d)!)}$ time.
\end{corollary}

Thus \textsf{$(k,c)$-SPC} on undirected graphs is in polynomial time whenever $d = k-c$ is constant; that is, it is in the complexity class $\mathsf{XP}$ with respect to the parameter $d$.
Prior to our work, no polynomial-time algorithm for this problem appears to have been known in this regime, even for simple cases such as \textsf{$(k,k-1)$-SPC} on undirected graphs. 

In contrast, our structural result for directed graphs, \Cref{thm:dirnew}, does not imply a faster algorithm for \textsf{$(k,c)$-SPC} in directed graphs. 
This is because \Cref{thm:dirnew} does not appear to imply a reduction from \textsf{$(k,c)$-SPC} to \textsf{$k$-DSP} in the manner of \cref{thm:undir}. 
Moreover, even if such a reduction did exist, this would not imply an algorithm for directed graphs analogous to \cref{cor}. 
This is because while \textsf{$k$-DSP} is polynomial-time solvable for constant $k$ in undirected graphs, it remains open whether even \textsf{3-DSP} over directed graphs can be solved in polynomial time.

\subsection{The Structure of Shortest Paths in Directed Graphs}\label{sec:struct}

Although our result for directed graphs, \cref{thm:dirnew}, does not appear to have immediate algorithmic applications, we believe it is still interesting from a graph theoretic perspective, especially in light of the current scarcity of results for shortest disjoint path problems in directed graphs.
In this section, we expand upon this idea with some remarks, and then state a lemma from our proof concerning the structure of shortest paths in directed graphs.

We currently have a poor understanding of the complexity of the \textsf{$k$-DSP} problem in directed graphs: for fixed $k\ge 3$, it is still not known if this problem is either polynomial-time solvable or \textsf{NP}-hard.
In fact, even the complexity of the seemingly easier \textsf{$(3,2)$-SPC} problem on directed graphs is open.
In this context, \cref{thm:dirnew} is compelling because it presents an example of interesting behavior which holds for collections of shortest paths in DAGs, and then continues to hold, under suitable generalization, for systems of shortest paths in general directed graphs. 
This sort of characterization appears to be rare in the literature.

More generally, the methods we use to establish \cref{thm:dirnew} involve combinatorial observations about the structure of shortest paths in directed graphs, and the interactions between them. 
We believe our analysis could offer more insight into resolving other problems that concern systems of shortest paths in directed graphs. 
There are many such problems, where the undirected case is well-understood, but in the directed case not much is known.
This barrier is in-part due to the relatively complex patterns of shortest paths which can appear in directed graphs. 
We hope that our analysis of directed shortest paths may shed light on problems for which the structural complexity of directed shortest paths is the bottleneck towards progress. 

One example of a problem where the undirected case is well-understood while the directed case remains poorly understood is, as discussed previously, the \textsf{$k$-DSP} problem.
Another curious example is the \textsf{Not-Shortest Path} 
problem, where the goal is to find an $st$-path that is not a shortest path.
Although \textsf{Not-Shortest Path} can be solved over undirected graphs in polynomial time \cite{DBLP:journals/ipl/KrasikovN04}, no polynomial time algorithm is known for this problem in directed graphs.
A third example is the problem of approximating the diameter of a graph. 
For undirected graphs there is an infinite hierarchy of algorithms that trade off between time and accuracy \cite{DBLP:conf/soda/CairoGR16}, while only two points on the hierarchy are known for directed graphs. 
Additionally, the roundtrip variant of diameter is the least understood of any studied variant of the diameter problem \cite{DBLP:journals/sigact/RubinsteinW19}. 
Another example is the construction of approximate hopsets: for directed graphs there is there a polynomial gap between upper and lower bounds \cite{DBLP:journals/siamdm/HuangP21,nicole-aaron-amazing-work,BodwinWinning}, while for undirected graphs the gap is subpolynomial \cite{DBLP:conf/spaa/ElkinN19, DBLP:journals/ipl/HuangP19}.
The preponderance of such examples motivates proving results like \Cref{thm:dirnew}, which characterize interesting behavior of shortest paths in directed graphs.

\paragraph*{Structural Lemma for Directed Shortest Paths}

One of the lemmas we establish on the way to proving \cref{thm:dirnew} is a general statement about the structure of shortest paths in directed graphs. It can be stated independently of the context of the proof of \cref{thm:dirnew} and
we highlight it here.

We categorize any shortest path $P$ into one of two main types, based on the ways that other shortest paths intersect with it. 
The following simple definition will be useful for defining our path types.

\begin{restatable}{definition}{defsp}\label{def:sp}
For any nonnegative integer $\ell$ and set of $\ell$ nodes, $v_1,v_2,\dots,v_\ell$, we say that the order $v_1 \rightarrow v_2 \rightarrow \dots \rightarrow v_\ell$ is a \emph{shortest-path ordering} if there is a shortest path containing all of the nodes $v_1,v_2,\dots,v_\ell$ in that order.
\end{restatable}

In addition to our two main path types, there is a third path type  which we call a \emph{trivial path} because it is easy to handle:

\begin{restatable}[Trivial Path]{definition}{deftriv}
Given a directed graph, nodes $a$ and $b$, and a shortest path $P$ from $a$ to $b$, we say $P$ is a \emph{trivial path} if $P$ contains at least one node $w$ such that $a\rightarrow w\rightarrow b$ is the \emph{only} shortest-path ordering of $a,w,b$.
\end{restatable}

Now we are ready to state our two main types of shortest paths. 
The first type is a \emph{reversing path}:

\begin{definition}[Reversing path] Given a directed graph, nodes $a$ and $b$, and a non-trivial shortest path $P$ from $a$ to $b$, $P$ is \emph{reversing} if $P$ contains at least one node $w$ such that $w$ falls on some shortest path from $b$ to $a$. A \emph{non-reversing path} is a non-trivial path that is not reversing. 
\end{definition}

We prove a lemma that characterizes the structure of reversing and non-reversing paths in terms of the possible shortest-path orderings of each node on the path and the endpoints of the path. See \cref{fig:F,fig:G} for a depiction of the structure enforced by the lemma.

\begin{restatable}[Reversing/Non-Reversing Lemma]{lemma}{lemrev}\label{lem:rev} Let $P$ be a non-trivial shortest path and let $a$ and $b$ be the first and last nodes of $P$ respectively. Then $P$ can be partitioned into three contiguous ordered segments with the following properties (where $a$ is defined to be in Segment 1, $b$ is defined to be in Segment 3, and Segment 2 could be empty).\\[3mm]
\underline{Segment 1} consists of nodes $w$ such that the shortest-path orderings of $a,w,b$ are precisely \\\centerline{$a\rightarrow w \rightarrow b$, and $b\rightarrow a \rightarrow w$.}\\[3mm]
\underline{Segment 2} consists of nodes $w$ such that the shortest-path orderings of $a,w,b$ are precisely \\
\[
\begin{cases}
\text{$a\rightarrow w \rightarrow b$, and $b\rightarrow w \rightarrow a$} & \text{if $P$ is a reversing path}\\
\text{$a\rightarrow w \rightarrow b$, and $b\rightarrow a \rightarrow w$, and $w\rightarrow b\rightarrow a$} & \text{if $P$ is a non-reversing path.}
\end{cases}\]\\[3mm]
\underline{Segment 3} consists nodes $w$ such that the shortest-path orderings of $a,w,b$ are precisely \\\centerline{$a\rightarrow w \rightarrow b$, and $w\rightarrow b\rightarrow a$.}
\end{restatable}

\begin{figure}[ht]
\centering
\begin{subfigure}[t]{0.49\textwidth}
\centering
\def\svgwidth{\linewidth}
\begingroup%
  \makeatletter%
  \providecommand\color[2][]{%
    \errmessage{(Inkscape) Color is used for the text in Inkscape, but the package 'color.sty' is not loaded}%
    \renewcommand\color[2][]{}%
  }%
  \providecommand\transparent[1]{%
    \errmessage{(Inkscape) Transparency is used (non-zero) for the text in Inkscape, but the package 'transparent.sty' is not loaded}%
    \renewcommand\transparent[1]{}%
  }%
  \providecommand\rotatebox[2]{#2}%
  \newcommand*\fsize{\dimexpr\f@size pt\relax}%
  \newcommand*\lineheight[1]{\fontsize{\fsize}{#1\fsize}\selectfont}%
  \ifx\svgwidth\undefined%
    \setlength{\unitlength}{283.46456693bp}%
    \ifx\svgscale\undefined%
      \relax%
    \else%
      \setlength{\unitlength}{\unitlength * \real{\svgscale}}%
    \fi%
  \else%
    \setlength{\unitlength}{\svgwidth}%
  \fi%
  \global\let\svgwidth\undefined%
  \global\let\svgscale\undefined%
  \makeatother%
  \begin{picture}(1,0.48)%
    \lineheight{1}%
    \setlength\tabcolsep{0pt}%
    \put(0,0){\includegraphics[width=\unitlength,page=1]{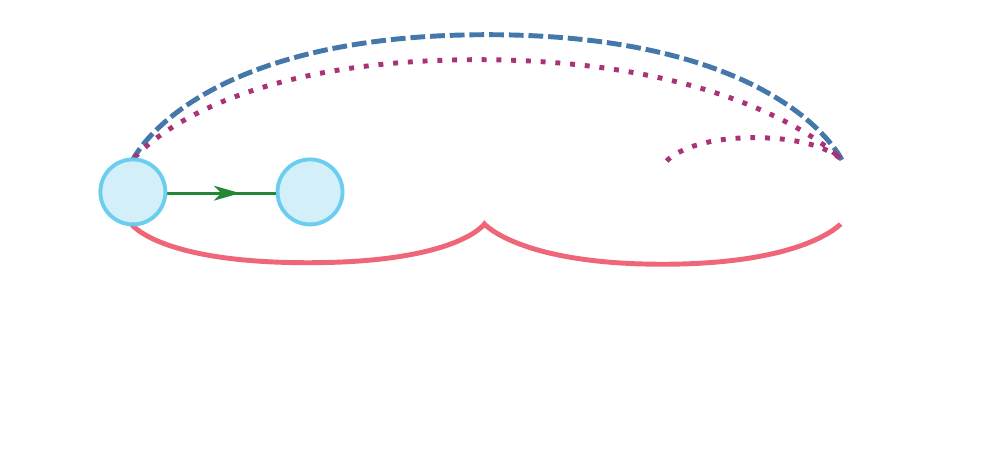}}%
    \put(0.1268755,0.27697765){\color[rgb]{0,0,0}\makebox(0,0)[lt]{\lineheight{1.25}\smash{\begin{tabular}[t]{l}$a$\end{tabular}}}}%
    \put(0,0){\includegraphics[width=\unitlength,page=2]{spc-F-smallertext.pdf}}%
    \put(0.8468734,0.27697765){\color[rgb]{0,0,0}\makebox(0,0)[lt]{\lineheight{1.25}\smash{\begin{tabular}[t]{l}$b$\end{tabular}}}}%
    \put(0,0){\includegraphics[width=\unitlength,page=3]{spc-F-smallertext.pdf}}%
    \put(0.16920879,0.13410252){\color[rgb]{0,0,0}\makebox(0,0)[lt]{\lineheight{1.25}\smash{\begin{tabular}[t]{l}Segment $1$\end{tabular}}}}%
    \put(0.41791737,0.13410252){\color[rgb]{0,0,0}\makebox(0,0)[lt]{\lineheight{1.25}\smash{\begin{tabular}[t]{l}Segment $2$\end{tabular}}}}%
    \put(0.6666259,0.13410252){\color[rgb]{0,0,0}\makebox(0,0)[lt]{\lineheight{1.25}\smash{\begin{tabular}[t]{l}Segment $3$\end{tabular}}}}%
  \end{picture}%
\endgroup%

\caption{{\bf Reversing Path.} The structure of a reversing path, as given by \cref{lem:rev}. The blue dashed path, pink solid path, and purple dotted path are examples of the allowed orderings for nodes in segments 1, 2, and 3 respectively.}
\label{fig:F}
\end{subfigure}
\hfill
\begin{subfigure}[t]{0.49\textwidth}
\centering
\def\svgwidth{\linewidth}
\begingroup%
  \makeatletter%
  \providecommand\color[2][]{%
    \errmessage{(Inkscape) Color is used for the text in Inkscape, but the package 'color.sty' is not loaded}%
    \renewcommand\color[2][]{}%
  }%
  \providecommand\transparent[1]{%
    \errmessage{(Inkscape) Transparency is used (non-zero) for the text in Inkscape, but the package 'transparent.sty' is not loaded}%
    \renewcommand\transparent[1]{}%
  }%
  \providecommand\rotatebox[2]{#2}%
  \newcommand*\fsize{\dimexpr\f@size pt\relax}%
  \newcommand*\lineheight[1]{\fontsize{\fsize}{#1\fsize}\selectfont}%
  \ifx\svgwidth\undefined%
    \setlength{\unitlength}{286.2992126bp}%
    \ifx\svgscale\undefined%
      \relax%
    \else%
      \setlength{\unitlength}{\unitlength * \real{\svgscale}}%
    \fi%
  \else%
    \setlength{\unitlength}{\svgwidth}%
  \fi%
  \global\let\svgwidth\undefined%
  \global\let\svgscale\undefined%
  \makeatother%
  \begin{picture}(1,0.47524752)%
    \lineheight{1}%
    \setlength\tabcolsep{0pt}%
    \put(0,0){\includegraphics[width=\unitlength,page=1]{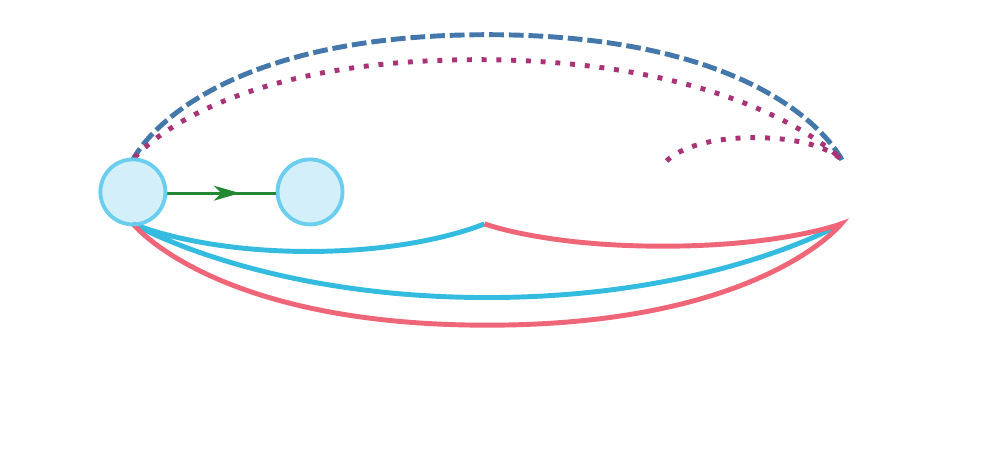}}%
    \put(0.1256193,0.2742353){\color[rgb]{0,0,0}\makebox(0,0)[lt]{\lineheight{1.25}\smash{\begin{tabular}[t]{l}$a$\end{tabular}}}}%
    \put(0,0){\includegraphics[width=\unitlength,page=2]{spc-G-smallertext.pdf}}%
    \put(0.83848851,0.2742353){\color[rgb]{0,0,0}\makebox(0,0)[lt]{\lineheight{1.25}\smash{\begin{tabular}[t]{l}$b$\end{tabular}}}}%
    \put(0,0){\includegraphics[width=\unitlength,page=3]{spc-G-smallertext.pdf}}%
    \put(0.16753346,0.0832731){\color[rgb]{0,0,0}\makebox(0,0)[lt]{\lineheight{1.25}\smash{\begin{tabular}[t]{l}Segment $1$\end{tabular}}}}%
    \put(0.41377957,0.0832731){\color[rgb]{0,0,0}\makebox(0,0)[lt]{\lineheight{1.25}\smash{\begin{tabular}[t]{l}Segment $2$\end{tabular}}}}%
    \put(0.66002564,0.0832731){\color[rgb]{0,0,0}\makebox(0,0)[lt]{\lineheight{1.25}\smash{\begin{tabular}[t]{l}Segment $3$\end{tabular}}}}%
    \put(0,0){\includegraphics[width=\unitlength,page=4]{spc-G-smallertext.pdf}}%
  \end{picture}%
\endgroup%

\caption{{\bf Non-Reversing Path.} The structure of a non-reversing path, as given by \cref{lem:rev}. The possible shortest-path orderings for nodes in segment 2 are represented by pink and light blue solid paths, while the orderings allowed for nodes in segments 1 and 3 are represented by blue dashed and purple dotted segments respectively.}
\label{fig:G}
\end{subfigure}
\caption{Possible orderings of vertices on shortest paths in the reversing and non-reversing cases. The blue circles are representative examples of the types of nodes on the path from $a$ to $b$ (in general this path will contain more than just five nodes).}
\end{figure}

In the proof of \cref{thm:dirnew} we employ the strategy of categorizing shortest paths as reversing or non-reversing (or trivial), and applying \cref{lem:rev} to glean some structure. 
We note, however,  that \cref{lem:rev} itself is not the main technical piece of the proof.

\section{Preliminaries}\label{sec:pre}

All graphs are assumed to have positive edge weights. 
Graphs are either undirected or directed, depending on the section. For any pair of nodes $(u,v)$, we use $\dist(u,v)$ to denote the shortest path distance from $u$ to $v$. 
Given a path $P$ and two nodes $u$ and $v$ occurring on $P$ in that order, we let $P[u,v]$ denote the subpath of $P$ with $u$ and $v$ as endpoints.  

For the \textsf{$(k,c)$-SPC} problem, we always let $d$ denote the difference $d=k-c$. When considering a particular solution to a \textsf{$(k,c)$-SPC} instance, we refer to the paths $P_1, \dots, P_k$ between $(s_1,t_1), \dots, (s_k,t_k)$ respectively as \emph{solution paths}. 
Any node in the graph which lies in $c$ of the solution paths is referred to as a \emph{max-congestion} node.

\subsection{Subpath Swapping}
\label{sec:swap}

In our proofs, we will frequently modify collections of shortest paths by ``swapping subpaths'' between intersecting paths. 
This procedure is depicted in \Cref{fig:A}, and we formally describe it below.
\begin{definition}[Subpath Swap]\label{defn:swap}
Let $\mathcal{R}$ be a collection of shortest paths in a directed graph.
Let $P$ and $Q$ be two paths in $\mathcal{R}$.
Let $a, b\in P\cap Q$ be nodes in these paths, such that $a$ appears before $b$ in both $P$ and $Q$.
We define \emph{swapping} the subpaths of $P$ and $Q$ between $a$ and $b$ to be updating the set of paths $\mathcal{R}$ by simultaneously replacing $P$ with $(P\setminus P[a,b])\cup Q[a,b]$ and $Q$ with $(Q\setminus Q[a,b])\cup P[a,b]$. 
We often refer to this process simply as ``swapping $P[a,b]$ and $Q[a,b]$.''
\end{definition}

\begin{figure}[ht]
\centering
\def\svgwidth{.9\linewidth}
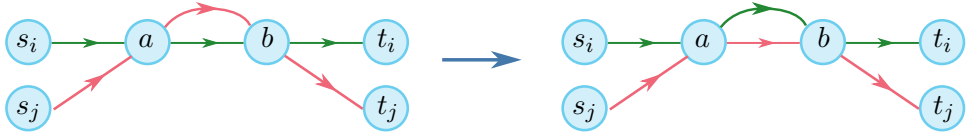
\caption{A simple subpath swap, where the subpaths from $a$ to $b$ of the green path (from $s_i$ to $t_i$) and pink path (from $s_j$ to $t_j$) are switched.}
\label{fig:A}
\end{figure}

\begin{observation}[Subpath Swap]
    \label{obs:subpath-swap}
    Let $\mathcal{R}$ be the solution to some \textsf{$(k,c)$-SPC} problem.
    Then swapping subpaths in $\mathcal{R}$ produces a new solution to the same \textsf{$(k,c)$-SPC} instance with the same set of max-congestion nodes.
\end{observation}
\begin{proof}
    This observation holds because
    swapping subpaths does not change the number of solution paths any given node is contained in, does not change the endpoints of any solution path, and all solution paths remain shortest paths.
\end{proof}

\subsection{Correspondence Between Our Results and \texorpdfstring{$(k,c)$}{(k,c)}-SPC}\label{sec:nuance}

In this section we detail some nuances regarding the correspondence between the statement of our results and the \textsf{$(k,c)$-SPC} problem. 
For the sake of generality and simplicity, we stated Theorems \ref{thm:dagnew}-\ref{thm:dirnew} independently of the \textsf{$(k,c)$-SPC} problem. 
In contrast, the original result of Amiri and Wargalla, corresponding to \cref{thm:dagnew}, was stated as follows:

\begin{restatable}[\cite{SDP-congestion-DAG}]{lemma}{lemdag}\label{lem:dag}
If $k > 3d$, then any instance of \textsf{$(k,c)$-SPC} on DAGs either has no solution, or has a solution where some solution path $P_i$ passes through all max-congestion nodes.
\end{restatable}

Although the statement of \cref{lem:dag} may initially seem unrelated to the statement of \cref{thm:dagnew}, their correspondence becomes clearer with the following observation, which is a simple generalization of an observation from \cite{SDP-congestion-DAG}:

\begin{observation}\label{lem:super-key}
Let $N$ be a positive integer.
Suppose $k > Nd$, and let $\mathcal{R}$ be a solution to an arbitrary \textsf{$(k,c)$-SPC} instance. Then for any set $S$ of $N$ max-congestion nodes in $\mathcal{R}$, there exists some solution path in $\mathcal{R}$ which contains every node in $S$.
\end{observation}
\begin{proof}
	By assumption, each max-congestion node lies on exactly $c = k-d$ of the solution paths of $\mathcal{R}$, and does \emph{not} lie on exactly $d$ of these solution paths.
	This means that the total number of solution paths of $\mathcal{R}$ avoided by at least one node in $S$ is $|S|d = Nd < k$.
	Consequently there exists at least one solution
	path passing through every node in $S$.
\end{proof}

To prove our results for the \textsf{$(k,c)$-SPC} problem in undirected graphs (\cref{thm:undir,cor}), we need to prove a lemma analogous to \cref{lem:dag} but for undirected graphs:

\begin{restatable}{lemma}{lmheavy}
\label{lm:heavy-path}
    If $k > 4d$, then any instance of \textsf{$(k,c)$-SPC} either has no solution, or has a solution where some solution path $P_i$ passes through all max-congestion nodes.
\end{restatable}
\noindent The following generalizes \Cref{thm:undirnew} and \Cref{lm:heavy-path}.

\begin{restatable}[General Undirected Result]{lemma}{thmundirgen}\label{thm:undirgen}

Given an undirected graph and subset $W$ of nodes, let $\mathcal{R}$ be a collection of shortest paths with the following property:

\begin{center}
\hfill for every set $S\subseteq W$ of 4 nodes, some path in $\mathcal{R}$ contains every node in $S$. \hfill $(\star)$
\end{center}
Further suppose that applying any sequence of $O(n^3)$ subpath swaps to $\mathcal{R}$ yields a collection of shortest paths that still has property $(\star)$. Then starting from $\mathcal{R}$, there exists a sequence of subpath swaps that results in a collection of shortest paths in which some path $P$ passes through all nodes in $W$.
\end{restatable}
We note that the quantity $O(n^3)$ is an unimportant technicality used in the proof of the following observation, and is chosen as a loose upper bound on the number of subpath swaps we will perform.

\begin{observation}
\label{obs:obv}
\cref{thm:undirgen} generalizes both \cref{thm:undirnew} and \cref{lm:heavy-path}.
\end{observation}

\begin{proof}
We begin by showing that \cref{thm:undirgen} generalizes \cref{thm:undirnew}. In the statement of \cref{thm:undirgen}, let $W$ be the entire vertex set. To see that the precondition of \cref{thm:undirnew} meets the precondition of \cref{thm:undirgen}, let $\mathcal{R}$ be a multiset of shortest paths containing the (perhaps exponentially large) set of all shortest paths in the graph each with multiplicity $Cn^3$ for a sufficiently large constant $C$. Then, after any sequence of $O(n^3)$ subpath swaps, $\mathcal{R}$ still contains at least one copy of every shortest path in the graph, so property $(\star)$ still holds. Finally, it is clear that the conclusion of \cref{thm:undirgen} implies the conclusion of \cref{thm:undirnew}.

Now we show that \cref{thm:undirgen} generalizes \cref{lm:heavy-path}. 
In the statement of \cref{thm:undirgen}, let $\mathcal{R}$ be a solution (if one exists) to the \textsf{$(k,c)$-SPC} instance  from the statement of \cref{lm:heavy-path} and let $W$ be the set of max-congestion nodes. 
To see that the precondition of \cref{lm:heavy-path} meets the precondition of \cref{thm:undirgen}, notice that by \cref{obs:subpath-swap}, after any sequence of subpaths swaps $\mathcal{R}$ is still a solution to the \textsf{$(k,c)$-SPC} instance with the same set of max-congestion nodes. Thus, by \cref{lem:super-key} property $(\star)$ holds after any sequence of subpath swaps applied to $\mathcal{R}$. Finally, it is clear that the conclusion of \cref{thm:undirgen} implies the conclusion of \cref{lm:heavy-path}.
\end{proof}

We also prove an analogue of \cref{lem:dag} and \cref{thm:undirgen} for directed graphs:

\begin{restatable}[General Directed Result]{lemma}{thmdirgen}\label{thm:dirgen}
Given a directed graph and subset $W$ of nodes, let $\mathcal{R}$ be a collection of shortest paths with the following property: 

\begin{center}
\hfill 
for every set $S\subseteq W$ of 11 nodes, some path in $\mathcal{R}$ contains every node in $S$.\hfill $(\dagger)$\end{center}
Further suppose that applying any sequence of $O(n^3)$ subpath swaps to $\mathcal{R}$ yields a collection of paths that still has property $(\dagger)$. Then starting from $\mathcal{R}$, there exists a sequence of subpath swaps that results in a collection of shortest paths in which either: \begin{enumerate} \item some path $P$ passes through all nodes in $W$, or \item the union of two paths $P,P'$ contain all nodes in $W$, and the first and last nodes on $P$ from $W$ are the same as the last and first nodes on $P'$ from $W$, respectively. \end{enumerate}
\end{restatable}
\noindent \cref{thm:dirgen} generalizes \cref{thm:dirnew}, in exactly the same way as  \cref{thm:undirgen} generalizes \cref{thm:undirnew} for undirected graphs.
In the same way that \cref{thm:undirgen} generalizes \cref{lm:heavy-path} for undirected graphs, \cref{thm:dirgen} implies the following lemma:

\begin{lemma}\label{thm:relax}
If $k > 11d$, then any instance of \textsf{$(k,c)$-SPC} on directed graphs either has no solution, or has a solution where the union of some two solution paths $P_i$ and $P_j$ contains all max-congestion nodes.
\end{lemma}

Although \cref{thm:relax} does not immediately lead to an algorithm for \textsf{$(k,c)$-SPC} on directed graphs, it specifies some structure which may be useful for future work on \textsf{$(k,c)$-SPC} and related problems.

One might wonder whether \cref{thm:relax} can be modified to have only one solution path $P_i$ that contains all max-congestion nodes, like for DAGs (\cref{lem:dag}) and undirected graphs (\cref{lm:heavy-path}), since this would imply interesting algorithms for \textsf{$(k,c)$-SPC} on directed graphs. 
Unfortunately, the answer to this question turns out to be no.
Similar to the counterexample against extending \cref{thm:undirnew} to directed graphs, we present a counterexample in \cref{app:counter} which rules out replacing two solution paths with a single solution path in \cref{thm:relax}.

\subsection{Organization} 
As suggested in the discussion from \cref{sec:nuance}, the goal of this work is to establish our general undirected result (\cref{thm:undirgen}) and our general directed result (\cref{thm:dirgen}). 
To that end, in \cref{sec:tech} we provide a technical overview of our results, then in \cref{sec:un} we prove \cref{thm:undirgen}, and in \cref{sec:dir} we prove \cref{thm:dirgen}.

In \cref{app:rel} we provide background on related work on disjoint paths problems. 
In \cref{app:counter} we provide a counterexample against extending \cref{lm:heavy-path} to directed graphs. 
Finally, in \cref{app:cor} we prove our corollary for undirected graphs (\cref{thm:undir}) given \cref{lm:heavy-path}, which is identical to Amiri and Wargalla's argument for DAGs, except with different constants.

\section{Technical Overview}\label{sec:tech}

\subsection{Prior Work on DAGs}
\label{subsec:prior}

As a starting point for our proofs for general graphs, we use Amiri and Wargalla's proof of \cref{thm:dagnew} for DAGs \cite{SDP-congestion-DAG}.
We describe their proof differently than they do for the sake of comparison to our work. 
We actually describe their proof of the following lemma (which they do not explicitly state), which is analogous to the statements of our general results (\cref{thm:undirgen,thm:dirgen}).

\begin{restatable}[\cite{SDP-congestion-DAG}]{lemma}{thmdaggen}\label{thm:daggen}
Given a DAG and subset $W$ of nodes, let $\mathcal{R}$ be a collection of shortest paths with the following property:

\begin{center}
\hfill
    for every set $S\subseteq W$ of 3 nodes, some path in $\mathcal{R}$ contains every node in $S$. \hfill $(\ast)$
\end{center}
Further suppose that applying any sequence of $O(n^3)$ subpath swaps to $\mathcal{R}$ yields a collection of paths that still has property $(\ast)$. Then starting from $\mathcal{R}$, there exists a sequence of subpath swaps that results in a collection of shortest paths in which some path $P$ passes through all nodes in $W$.
\end{restatable}

The proof of \cref{thm:daggen} is as follows. Let $a$ and $b$ be the first and last nodes in $W$ (respectively) in a topological ordering of the DAG. Let  $v_1,\dots v_{|W|-2}$ be the remaining nodes in $W$ in order topologically.  
For ease of notation, we consider $\mathcal{R}$ to be changing over time via subpath swaps, and we will let $\mathcal{R}$ denote the current value of $\mathcal{R}$. 

The argument is inductive. For the base case, by property $(\ast)$ there is a path in $ \mathcal{R}$ that contains $a$ and $b$. Suppose inductively that a path $P\in \mathcal{R}$ currently contains $a$, $b$, and $v_1,\dots,v_{\ell}$ for some $\ell$. We would like to perform a subpath swap to augment $P$ by adding $v_{\ell+1}$ to $P$. To do so, we consider a path $P'\in \mathcal{R}$ that contains $v_{\ell}$, $v_{\ell+1}$, and $b$, where such a path exists by property $(\ast)$. Importantly, because the graph is a DAG, $v_{\ell}$, $v_{\ell+1}$, and $b$ appear in that order on both $P$ and $P'$. 
Thus, according to the definition of a subpath swap (\cref{defn:swap}) we can swap $P[v_{\ell},b]$  with $P'[v_{\ell},b]$, as shown in \Cref{fig:B}.
As a result of this subpath swap, $P$ now contains $v_{\ell+1}$ in addition to all of the nodes in $W$ that $P$ originally contained.
By induction, this completes the proof.

\begin{figure}[ht]
\centering
\def\svgwidth{.65\linewidth}
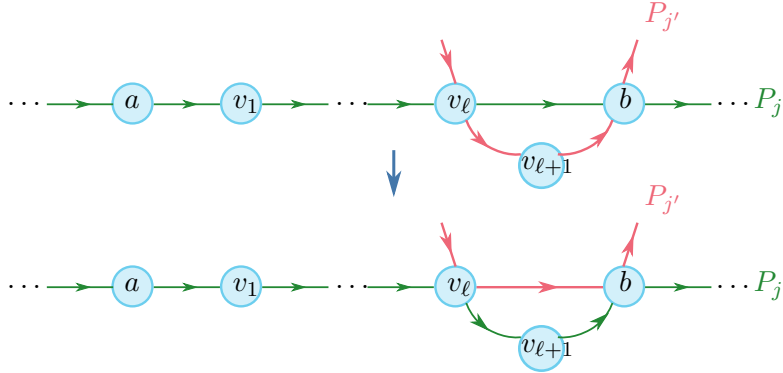
\caption{In a DAG, the topological ordering of the nodes allows us to perform a sequence of subpath swaps, each adding the next node in $W$ in order to a path $P$.}
\label{fig:B}
\end{figure}

\subsection{Undirected Graphs}

Recall that our goal for undirected graphs is to prove the following theorem:

\thmundirgen*

The essential property that enables the subpath swapping in the above argument for DAGs is the fact that for any triple of nodes in a DAG, there is \emph{only one} possible order this triple can appear on any path. 
This property is not exactly true for general undirected graphs, but we observe that a similar ``consistent ordering'' property is true: if there is a \emph{shortest} path containing nodes $u,v,w$ in that order, then any shortest path containing these nodes, has them in that order (or in the reverse order $w,v,u$, but since the graph is undirected we can without loss of generality assume they are in the order $u,v,w$). This property is true simply because $\dist(u,w)$ is larger than both $\dist(u,v)$ and $\dist(v,w)$, so $v$ must appear between $u$ and $w$ on any shortest path. 

To perform subpath swapping on undirected graphs, we need an initial pair of nodes in $W$ such that the rest of the nodes in $W$ will be inserted between this initial pair (in the DAG algorithm, this initial pair $a,b$ was the first and last nodes in $W$ in the topological order).
For undirected graphs, our initial pair is the pair $a,b$ of nodes in $W$ whose distance is maximum. We order the rest of the nodes in $W$ by their distance from $a$, to form $v_1,\dots,v_{|W|-2}$. 

Now, our consistent ordering property from above implies the following: for any shortest path $P$ containing $a$, all nodes in $W\cap P$ are ordered as a subsequence of $a,v_1,v_2,\dots,v_m,b$ on $P$. As a result, we can perform the same type of iterative subpath swapping argument as the DAG algorithm.

\subsection{Roundtrip Paths in Directed Graphs}

The situation for directed graphs is significantly more involved than the previous cases. 
There are several challenges that are present for directed graphs that were not present for either undirected graphs or DAGs. 
These difficulties stem from the fact that the \emph{interactions} between shortest paths is much more complicated in directed graphs than in DAGs or undirected graphs. 

We first outline these challenges, and then provide an overview of how we address them. 
Our techniques for addressing these issues exemplify that despite the possibly complex arrangement of shortest paths in directed graphs, there still exists an underlying structure to extract. 
We hope that our methods might illuminate some structural properties of shortest paths in directed graphs in a way that could apply to other directed-shortest-path problems.

Recall that our goal is to prove the following theorem:

\thmdirgen*

\subsubsection*{Challenges for directed graphs}

\paragraph*{Challenge 1: No ``extremal'' nodes.} In the proofs for DAGs and undirected graphs, to begin building the path $P$ containing all nodes in $W$, we chose two initial extremal nodes $a,b\in W$ and added the rest of the nodes of $W$ between $a$ and $b$. These nodes $a,b$ were straightforward to choose because there was only one pair of nodes in $W$ that could possibly appear first and last on a path containing all nodes in $W$ 
(for DAGs $a$ and $b$ were the first and last  nodes in a topological ordering of $W$, and for undirected graphs $a$ and $b$ were the pair of nodes in $W$ with largest distance). 

For directed graphs however, it is entirely unclear how to pick these two initial extremal nodes. For instance, choosing the pair of nodes $a,b\in W$ with largest directed distance $\dist(a,b)$ does not work because there could be a shortest path $Q$ containing $a$ and $b$, where $a$ and $b$ are not the first and last nodes in $W$ on $Q$ (in particular, if $b$ appears before $a$ on the shortest path). 

To circumvent this issue for directed graphs, we avoid selecting a pair of initial nodes at all. 
Without initial nodes as an anchor, we cannot build our path $P$ in order from beginning to end as we did for DAGs and undirected graphs. Instead, our goal is to ensure the following weaker property: 
as we iteratively transform the overall collection of shortest paths $\mathcal{R}$, the set of nodes in $W$ on the path in $\mathcal{R}$ containing the most nodes in $W$ grows over time. 
That is, the path of $\mathcal{R}$ containing the most nodes in $W$ might currently be $P$, but at the previous iteration, the path of $\mathcal{R}$ with the most nodes in $W$ might have been a different path $P'$. 
In this case, our weaker property ensures that $P$ currently contains a \emph{superset} of the nodes in $W$ that $P'$ contained at the previous iteration.

To perform a single iteration with this guarantee, we may need to significantly reconfigure \emph{many} different paths of $\mathcal{R}$ via many subpath swaps. 
As a result, our path building procedure is much more intricate than the procedures employed for DAGs and undirected graphs. 

\paragraph*{Challenge 2: No consistent ordering of nodes on shortest paths.} 
In the proof for DAGs and undirected graphs, we were able to perform subpath swaps due to the following crucial property: consider an arbitrary set of nodes $v_1,\dots, v_{\ell}$ (for any $\ell$) in a DAG or an undirected graph. If there is a shortest path containing the nodes $v_1,\dots, v_\ell$ in that order, then \emph{every} shortest path containing these nodes has them in that same order.

For directed graphs, however, this property is not even close to being true. In fact, given that the nodes $v_1,\dots, v_\ell$ appear in that order on some shortest path, there are \emph{exponentially many} other possible orderings of these nodes on other shortest paths. For instance, when $\ell=4$, given that the nodes $v_1,v_2,v_3,v_4$ appear in that order on some shortest path, there are eight possible orderings of these nodes on shortest paths, as depicted in \Cref{fig:C} (note that despite the large number of possible orderings, not all orderings are possible; for instance the ordering $v_1,v_2,v_4,v_3$ is not possible, as this would imply that $\dist(v_1,v_4)<\dist(v_1,v_3)$, which contradicts our assumption that some shortest path contains $v_1,v_2,v_3,v_4$ in that order).

\begin{figure}[ht]
\centering
\def\svgwidth{\linewidth}
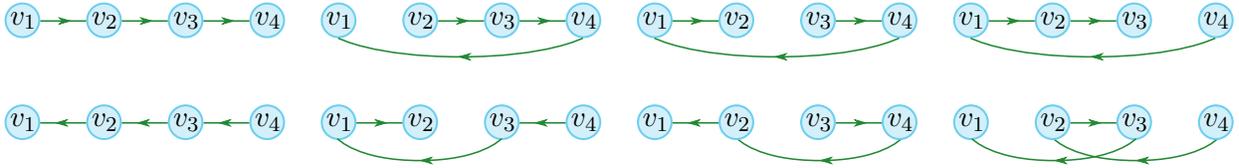
\caption{If nodes $v_1, v_2, v_3, v_4$ appear in that order on some shortest path, they can still appear on different shortest paths in up to seven other distinct possible orders.}
\label{fig:C}
\end{figure}

Because directed graphs have no consistent ordering of nodes on shortest paths, it becomes much more difficult to perform subpath swaps like those in the algorithms for DAGs and undirected graphs.

To address this challenge, we provide a structural analysis of the ways in which shortest paths in directed graphs can interact with each other. First, as introduced in \cref{sec:struct}, we categorize shortest paths into two main types, \emph{reversing} paths and \emph{non-reversing} paths, and we prove the Reversing/Non-Reversing Path Lemma (\cref{lem:rev}).
Roughly speaking, this lemma is useful because it helps us construct sets of nodes that exhibit some sort of consistent ordering property. This, in turn, enables us to perform sequences of subpath swaps.
Defining these consistently ordered sets of nodes and the corresponding subpath swaps is the most involved part of the proof, and works differently for each of the two path types.

\subsubsection*{Proof structure} 
Our proof is structured as follows. Initially, we define $P$ to be the path in $\mathcal{R}$ that contains the most nodes in $W$ (breaking ties arbitrarily).
Then we proceed with the following two cases:
\begin{itemize}
    \item (Case 1) We first check whether, roughly speaking, $P$ is contained in a \emph{cycle} that contains \emph{all} nodes in $W$. In this case, we can use a sequence of subpath swaps to build a second path $P'$ so that the union of $P$ and $P'$ contains all nodes in $W$ and have the ``roundtrip'' structure specified in the theorem statement, in which case we are done.
    \item (Case 2) If we are not in Case 1, our goal is to augment some path in $\mathcal{R}$ so that the set of nodes of $W$ on the path in $\mathcal{R}$ with the most nodes in $W$ grows (the goal introduced in the discussion of Challenge 1).
\end{itemize}  
After going through these cases, if we are not done we redefine $P$ as the path in $\mathcal{R}$ containing the most nodes in $W$ and repeatedly apply the appropriate case, until we are done.
Most technical details of our proof are in the path augmentation procedure of Case 2.
We elaborate on the main ideas for this procedure next.

\subsubsection*{Handling Case 2: Path Augmentation} Recall that our goal is the following: Let $P$ be the path in $\mathcal{R}$ containing the most nodes in $W$, and let $W'=W\cap P$. Fix a node $u\in W\setminus W'$. We wish to perform subpath swaps to yield a path $P'$ that contains every node in $W'\cup \{u\}$.

We begin with a few warm-up cases to motivate our general approach. 
\paragraph*{Warm-up cases}

Let $a$ and $b$ be the first and last nodes on $P$ that are in $W$, respectively. By property $(\dagger)$ there exists a path in $\mathcal{R}$ containing $a$, $b$, and $u$. We will suppose in all of the subsequent examples, that for \emph{all} paths in $\mathcal{R}$ containing the nodes $a$, $b$, and $u$, the node $u$ is always the last node in $W$ on the path; this is not a conceptually important assumption and it makes the description simpler. 

We claim that if there is a path $P'\in \mathcal{R}$ containing $a$, $b$, and $u$, such that $a$ appears before $b$, then we are done. 
Indeed, as depicted in \Cref{fig:D}, in this case we can simply swap the subpath $P[a,b]$ with $P'[a,b]$,
to form a new solution where $P'$ contains $W'\cup \{u\}$. 

\begin{figure}[ht]
\centering

\begin{subfigure}{0.43\textwidth}
\centering
\def\svgwidth{\linewidth}
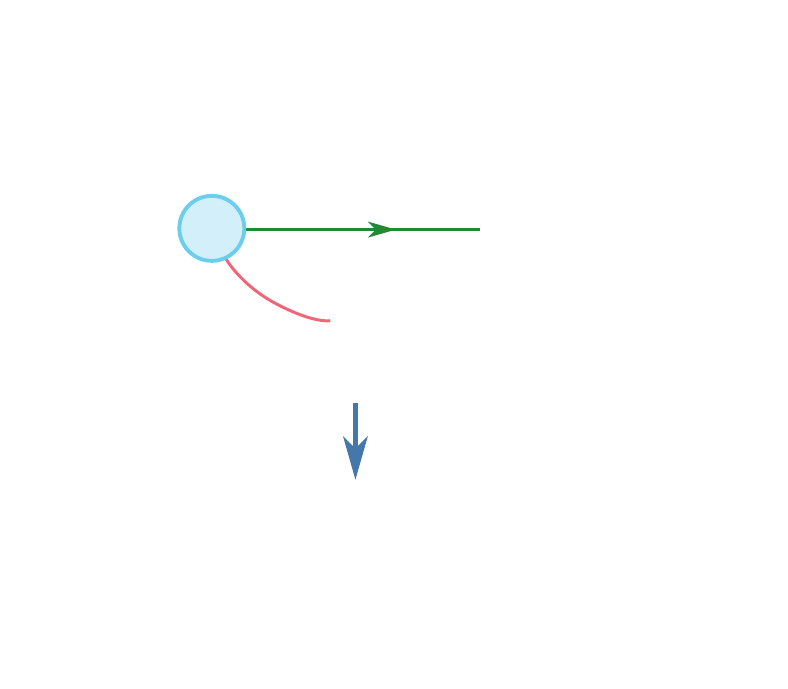
\caption{If $a$ appears before $b$ on $P'$, a subpath swap lets $P'$ pass through $u$ together with all of the nodes in $W'$.}
\label{fig:D}
\end{subfigure}
\hfill
\begin{subfigure}{0.54\textwidth}
\centering
\def\svgwidth{\linewidth}
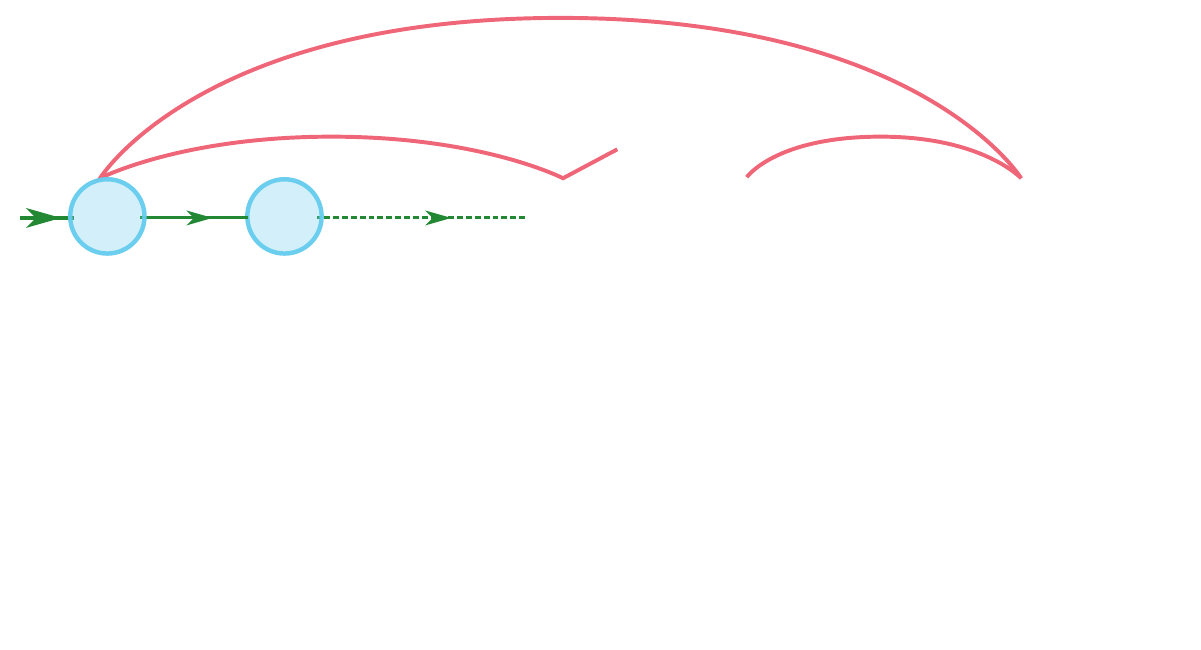
\caption{More generally, when $P'$ passes through $v_{\ell+1}, b, a, v_\ell$ in that order, we can perform two subpath swaps to get $P'$ to pass through the rest of $W'$. 
Here, the dotted segments indicate portions of the paths which pass through the nodes of $W'$ that are not labeled in the figure.}
\label{fig:E}
\end{subfigure}
\caption{The two warm-up cases.}
\end{figure}

Now we will slightly generalize this warm-up case. Let $a,v_1,v_2,\dots,b$ be the nodes in $W'$ in the order they appear on $P$. Consider $v_\ell$ and $v_{\ell+1}$ for any $\ell$. By property $(\dagger)$ there exists a path in $\mathcal{R}$ containing $a$, $b$, $u$, $v_\ell$, and $v_{\ell+1}$. We know from the previous warm-up case that if there is a pathin $\mathcal{R}$ containing these nodes such that $a$ appears before $b$, then we are done. We also claim that if there exists a path $P'\in\mathcal{R}$ containing these nodes such that $v_{\ell+1}, b,a, v_\ell$ appear in that order, then we are done. 
This is because as depicted in \Cref{fig:E}, we can swap the subpath $P[a,v_{\ell}]$ with $P'[a,v_{\ell}]$, and swap the subpath $P[v_{\ell+1},b]$ with $P'[v_{\ell+1},b]$.
Now, $P'$ contains $W'\cup \{u\}$.

\paragraph*{General Approach: ``Critical nodes''}

We will motivate our general approach in the context of the above warm-up cases. In the second warm-up case, we considered 5 nodes  ($a$, $b$, $u$, $v_\ell$, and $v_{\ell+1}$) in $W'$, and argued that if there is a path $P'\in \mathcal{R}$ containing these 5 nodes in one of several ``good'' orders, then we are done because we can perform subpath swaps to reroute $P'$ through all of $W'\cup \{u\}$. 
Thus, our goal is to choose these 5 (or in general, at most $11$) nodes carefully, to guarantee that they indeed fall into a ``good'' order on some path in $\mathcal{R}$. We refer to these at most $11$ nodes as \emph{critical nodes}:

\begin{definition}[Critical Nodes (Informal)] 
Given a path $P\in \mathcal{R}$, a set of nodes $T\subseteq W$ of size $|T|\leq 11$ are \emph{critical nodes of $P$} if there exists a path $P'\in\mathcal{R}$ containing the nodes of $T$ in an order that allows us to perform subpath swaps to reroute some path in $\mathcal{R}$ through all of $W'\cup \{u\}$ (where $W'=W\cap P$).
\end{definition}

\emph{Our general approach is to show that any path $P\in \mathcal{R}$ contains a set of at most 11 critical nodes.} We remind the reader that this section only concerns Case 2, so our goal is to show that $P$ contains a set of critical nodes only if $P$ does not already fall into Case 1. After showing that $P$ contains a set of critical nodes, we are done, because performing subpath swaps to yield a path containing $W'\cup \{u\}$ was our stated goal. 

It is not at all clear a priori that any $P$ should contain a set of critical nodes. Indeed, the critical nodes need to be chosen very carefully. Specifically, they need to be chosen based on the structure of the path $P$. 

This is where the definitions from \cref{sec:struct} come into play.  We categorize $P$ based on whether it is \emph{reversing} or \emph{non-reversing}, (or trivial). Then we construct the critical nodes of $P$ using a different procedure specialized for which type of path $P$ is.

If $P$ is a reversing path, then we argue that a valid choice of critical nodes are $a$, $b$, and $u$, along with the two nodes at the two boundaries between the segments defined in \cref{lem:rev} (and a few other nodes for technical reasons). To make this argument, which we will not detail here, we construct an involved series of subpath swaps to reroute some path through all of $W'\cup \{u\}$.

On the other hand, if $P$ is a non-reversing path, the critical nodes are less straightforward to define than if $P$ is a reversing path. Indeed, defining the critical nodes for non-reversing paths is the most conceptually difficult part of our proof. The high-level reason for this difficulty  is the fact that the nodes of Segment 2 of a non-reversing path are quite unconstrained because they admit 3 possible shortest-path orderings instead of only 2. 

To be more concrete, suppose every node on $P$ falls into Segment 2. For simplicity, suppose we were to choose critical nodes $a,b,w$, where $a$ and $b$ are the endpoints of $P$ and also happen to be in $W'$, and $w$ is any other node in $W'$. Consider the path $P'\in\mathcal{R}$ containing $a$, $b$, and $w$, which exists by property $(\dagger)$. By the definition of Segment 2, there are 3 possible orderings of $a,b,w$ on $P'$ ($a\rightarrow w \rightarrow b$, or $b\rightarrow a \rightarrow w$, or $w\rightarrow b\rightarrow a$). Note that these 3 orderings are cyclic shifts of one another. Suppose, as an illustrative example, that $P'$ has the ordering $b\rightarrow a \rightarrow w$. We would like to reroute $P'$ through all of $W'\cup \{u\}$, but we have a problem. For any node $s\in W'$ that appears between $a$ and $w$ on $P$, we can reroute $P'$ through $s$ by swapping the subpath $P[a,w]$ with $P'[a,w]$; however, for a node $s\in W'$ that appears between $w$ and $b$ on $P$, we cannot do this because there is no path segment from $w$ to $b$ on $P'$, since $b$ appears before $w$ on $P'$. Thus, we cannot reroute $P'$ through $s$. That is, no matter how we choose the critical nodes, if we do not impose extra structural constraints, we can always identify a segment of the path that we cannot reroute $P'$ through. Overcoming this issue is our main technical challenge, and we defer it to the full proof.

\section{Undirected Graphs}\label{sec:un}

As detailed in \cref{sec:pre}, our goal for undirected graphs is to prove the following lemma:

\thmundirgen*

\begin{proof}

	For ease of notation, we consider $\mathcal{R}$ to be changing over time via subpath swaps, and we will let $\mathcal{R}$ denote the current value of $\mathcal{R}$. 
	
	Among all nodes in $W$, select two nodes $a,b$ which maximize $\dist(a,b)$. 
	Let $v_1, v_2, \dots, v_{|W|-2}$ be the remaining nodes of $W$, ordered by increasing distance from the node $a$. (For the sake of this definition we break ties arbitrarily, but in reality, there cannot be any ties.)
	
	We now inductively show that for any positive integer $\ell$, there is a sequence of subpath swaps that yields a path that
	passes through the nodes $a, v_1, \dots, v_\ell$, and $b$, in that order.
	
		For the base case where $\ell = 1$, $\mathcal{R}$ contains a path passing through $a, b$, and $v_1$ by property $(\star)$.
	In this path, $v_1$ must be between $a$ and $b$ since $a$ and $b$ have maximum distance among all pairs of nodes in $W$.
(Since the graph is undirected we can define the ``beginning'' of the path to be $a$ to get a path through $a, v_1$, and $b$, in that order.)

	For the inductive step, let $\ell$ be a fixed positive integer with $1 < \ell < |W|$.
	By the inductive hypothesis, we may assume there is a path $P\in\mathcal{R}$ that passes through $a, v_1, \dots, v_\ell$, and $b$ in that order, for some fixed positive integer $\ell$.
	Our goal is to show that we can perform some subpath swaps to $\mathcal{R}$, to produce a path which passes through $a, v_1, \dots, v_{\ell+1}$, and $b$ in that order. 

	By property $(\star)$, some path $Q\in\mathcal{R}$ must pass through $a, b, v_\ell$, and $v_{\ell+1}$.
	
		By our choice of $a$ and $b$, we know that $v_{\ell}$ and $v_{\ell+1}$ must be between $a$ and $b$ on the path $Q$. 
	Since we ordered the $v_i$ nodes by distance from $a$, it follows that $P$ contains $a, v_\ell, v_{\ell+1}$, and $b$ in that order.
		Now we swap the subpath $P[v_\ell,b]$ with $Q[v_\ell,b]$.
	After the swap, path $P$ passes through the node  $v_{\ell+1}$ in addition to $a, v_1, \dots, v_\ell$, and $b$, in the desired order. 
    This completes the induction, and proves the result.
\end{proof}

\section{Directed Graphs}
\label{sec:dir}

As detailed in \cref{sec:pre}, our goal for directed graphs is to prove the following lemma:

\thmdirgen*

The section is organized as follows. In \cref{subsec:struct}, we
prove the Reversing/Non-Reversing Path Lemma introduced in the introduction (\cref{lem:rev}).
In \cref{subsec:swap} we prove a ``critical node'' lemma that specifies conditions under which we can perform subpath swaps on directed graphs. 
  In \cref{subsec:lem}, we complete the proof of \cref{thm:dirgen}.

  \subsection{Structural Lemma for Reversing/Non-Reversing Paths}\label{subsec:struct}

In this section, we will prove \cref{lem:rev}, which was introduced in the introduction. We actually prove something slightly more general than \cref{lem:rev}, so the lemma statement and some of the definitions are slightly different here than in the introduction.

First, we recall the definition of a shortest-path ordering (which is the same as before):

\defsp*

The following simple lemma about shortest-path orderings will be useful.

\begin{lemma}\label{lem:sp}
For any triple $\set{u,v,w}$ of nodes, if $u\rightarrow v\rightarrow w$ is a shortest-path ordering, then neither $v\rightarrow u\rightarrow w$ nor $u\rightarrow w\rightarrow v$ is a shortest-path ordering.
\end{lemma}
\begin{proof}
    Suppose $u\rightarrow v\rightarrow w$ is a shortest-path ordering. This means that $\dist(u,v)<\dist(u,w)$ and $\dist(v,w)<\dist(u,w)$. If $v\rightarrow u\rightarrow w$ were a shortest-path ordering, this would imply that $\dist(v,w)>\dist(u,w)$, a contradiction. If $u\rightarrow w\rightarrow v$ were a shortest-path ordering, this would imply that $\dist(u,v)>\dist(u,w)$, a contradiction.
\end{proof}

Now we define trivial paths, reversing paths, and non-reversing paths in a slightly more general way than in the introduction.

\begin{definition}[Trivial Path] Given a directed graph, nodes $a,b$, and a shortest path $P$ from $a$ to $b$, $P$ is a \emph{trivial path} with respect to a node $w$ on $P$ if $a\rightarrow w\rightarrow b$ is the \emph{only} shortest-path ordering of $a,w,b$.
\end{definition}

\begin{definition}[Reversing path] Suppose we are given a directed graph, nodes $\set{a,b}$, a shortest path $P$ from $a$ to $b$, and a subset $W'$ of the nodes on $P$ such that $P$ is not a trivial path with respect to any node in $W'$. Then, $P$ is a \emph{reversing path} with respect to $W'$ if $P$ contains at least one node $w\in W'$ such that $w$ falls on some $b$-to-$a$ shortest path. Otherwise, $P$ is a \emph{non-reversing} path with respect to $S$.
\end{definition}

We call the subset of nodes $W'$ intentionally, to hint that when we apply this definition to prove \cref{thm:dirgen}, $W'$ will be set to $P\cap W$.

Now, we state our structural lemma for reversing and non-reversing paths, which is slightly more general than \cref{lem:rev}:

\begin{restatable}[Reversing/Non-Reversing Path Lemma]{lemma}{lemcombo}
\label{lem:combo}
Suppose we are given a directed graph, nodes $\set{a,b}$, a shortest path $P$ from $a$ to $b$, and a subset $W'$ of the nodes on $P$ such that $P$ is not a trivial path with respect to any node in $W'$. Then, $P$ can be partitioned into three contiguous segments with the following properties (where $a$ is defined to be in Segment 1, $b$ is defined to be in Segment 3, and Segment 2 could be empty).\\[3mm]
\underline{Segment 1:} nodes $w$ such that if $w\in W'$ then the shortest-path orderings of $a,w,b$ are precisely \\\centerline{$a\rightarrow w \rightarrow b$, and $b\rightarrow a \rightarrow w$.}\\[3mm]
\underline{Segment 2:} nodes $w$ such that if $w\in W'$ then the shortest-path orderings of $a,w,b$ are precisely \\
\[
\begin{cases}
\text{$a\rightarrow w \rightarrow b$, and $b\rightarrow w \rightarrow a$} & \text{if $P$ is a reversing path wrt $W'$}\\
\text{$a\rightarrow w \rightarrow b$, and $b\rightarrow a \rightarrow w$, and $w\rightarrow b\rightarrow a$} & \text{if $P$ is a non-reversing path wrt $W'$.}
\end{cases}\]\\[3mm]
\noindent\underline{Segment 3:} nodes $w$ such that if $w\in W'$ then the shortest-path orderings of $a,w,b$ are precisely \\\centerline{$a\rightarrow w \rightarrow b$, and $w\rightarrow b\rightarrow a$.}
\end{restatable}

\begin{proof}
We first note that for any node $w$ on $P$, $a\rightarrow w \rightarrow b$ is a shortest-path ordering, since $P$ is a shortest path from $a$ to $b$. 

Given an arbitrary node $w\in W'$, we will enumerate the possible shortest-path orderings for $a,w,b$ besides $a\rightarrow w \rightarrow b$. First, $a\rightarrow w \rightarrow b$ is the only shortest-path ordering of $a,w,b$ that starts with $a$, simply because $\dist(a,w) < \dist(a,b)$. 
Similarly, $a\rightarrow w \rightarrow b$ is the only shortest-path ordering of $a,w,b$ that ends with $b$, simply because $\dist(w,b) < \dist(a,b)$. 
The remaining possible orderings are (1) $b\rightarrow a \rightarrow w$, (2) $w\rightarrow b\rightarrow a$, and (3) $b \rightarrow w\rightarrow a$.

By \cref{lem:sp}, if $w$ admits ordering 3, then $w$ cannot admit either of the first two orderings.

Let $w_1$ be the last node on $P$ that admits the ordering 1 ($b\rightarrow a \rightarrow w_1$) and \emph{not} ordering 2 ($w_1\rightarrow b \rightarrow a$), if such a node exists. Define Segment 1 as the prefix of $P$ up to an including $w_1$. We claim that if $w$ is in Segment 1 then $w$ also admits ordering 1 and \emph{not} ordering 2. Since $P$ is a shortest path, $a\rightarrow w\rightarrow w_1\rightarrow b$ is a shortest-path ordering. Combining this with $b\rightarrow a \rightarrow w_1$, we have that $b\rightarrow a \rightarrow w \rightarrow w_1$ is a shortest-path ordering. Thus, $b\rightarrow a \rightarrow w$ is a shortest-path ordering, so $w$ admits ordering 1, as desired. Now for the contrapositive of the second part of our claim, we will show that if $w$ admits ordering 2 ($w\rightarrow b \rightarrow a$), then so does $w_1$. This is by a similar argument: Combining  $w\rightarrow b \rightarrow a$ with $a\rightarrow w\rightarrow w_1\rightarrow b$ yields $w\rightarrow w_1 \rightarrow b \rightarrow a$, and thus $w_1\rightarrow b \rightarrow a$.

Now we claim that if $w$ is in Segment 1 then $a,w,b$ have no other shortest-path orderings besides $a\rightarrow w \rightarrow b$ and $b\rightarrow a \rightarrow w$. This is simply due to the above enumeration of possible shortest-path orderings: We know that $w$ admits ordering 1, we have already ruled out ordering 2, and we have already shown that ordering 3 cannot coexist with ordering 1. Thus, we have shown that the shortest-path orderings of $a,w,b$ are precisely $a\rightarrow w \rightarrow b$, and $b\rightarrow a \rightarrow w$, as desired.

Let $w_2$ be the first node on $P$ that admits ordering  2 ($w_2 \rightarrow b\rightarrow a$) and \emph{not} ordering 1 ($b \rightarrow a \rightarrow w_2$),  if such a node exists. Define Segment 3 as the suffix of $P$ starting at $w_2$. An argument that is exactly symmetric to the above argument for Segment 1 proves that every node $w$ in Segment 3 has precisely the shortest-path orderings $a\rightarrow w \rightarrow b$, and $w\rightarrow b\rightarrow a$, as desired.

Let $w_2$ be the first node on $P$ that admits ordering  2 ($w_2 \rightarrow b\rightarrow a$) and \emph{not} ordering 1 ($b \rightarrow a \rightarrow w_2$),  if such a node exists. Define Segment 3 as the suffix of $P$ starting at $w_2$. An argument that is exactly symmetric to the above argument for Segment 1 proves that every node $w$ in Segment 3 has precisely the shortest-path orderings $a\rightarrow w \rightarrow b$, and $w\rightarrow b\rightarrow a$, as desired.

Now, suppose $w$ is in Segment 2 i.e. $w$ lies between $w_1$ and $w_2$. By the definition of $w_1$, we know that $w$ either admits both orderings 1 and 2, or $w$ does not admit ordering 1. By the definition of $w_2$, we know that $w$ either admits both orderings 1 and 2, or $w$ does not admit ordering 2. Putting these together, $w$ either admits \emph{both} or \emph{neither} of orderings 1 or 2. If $w$ admits both orderings, then we have shown that $w$ cannot admit ordering 3, and thus $w$ does not admit any additional shortest-path orderings besides $a\rightarrow w\rightarrow b$. In this case, $w$ satisfies the non-reversing path case of Segment 2. On the other hand, if $w$ admits neither ordering 1 nor 2, then because $P$ is not a trivial path with respect to $W'$, $w$ must admit ordering 3. In this case, $w$ satisfies the reversing path case of Segment 2. 

It remains to show that $P$ cannot contain both a node $w$ that satisfies the reversing path case of Segment 2, and a node $w'$ that satisfies the non-reversing path case of Segment 2. This is simple to show: If $w'$ appears before $w$ on $P$, then $w'\rightarrow w\rightarrow b$ is a shortest-path ordering. By definition, $w'\rightarrow b\rightarrow a$ is also a shortest-path ordering. Combining these, we get the shortest-path ordering $w'\rightarrow w\rightarrow b \rightarrow a$, which implies the shortest-path ordering $w\rightarrow b \rightarrow a$, a contradiction to the definition of $w$. If $w'$ appears after $w$ on $P$, then a symmetric argument applies. This completes the proof.
\end{proof}

\subsection{Critical Node Lemma}
\label{subsec:swap}

In this section, we establish a lemma that specifies conditions under which we can perform subpath swaps on directed graphs. 
The intuition behind the lemma is motivated by our general approach of identifying ``critical nodes,'' outlined previously in the technical overview.

In the following lemma, the nodes $T$ are the critical nodes.
Although in the technical overview the critical nodes were said to come from a specific path $P$, no such path is mentioned in the statement of the lemma -- instead, such a path will be involved when the lemma is applied.

\begin{lemma}[Critical Node Lemma]
    \label{lem:reroute}
   Assume the premise of \cref{thm:dirgen}. 
   Suppose additionally we are given nodes $a_1, \dots, a_r, b_1, \dots, b_r$, a subset of nodes $W^*\subseteq W$, and a further subset of nodes 
   \[T=\left(\set{a_1, \dots, a_r, b_1, \dots, b_r}\cup U\right)\subseteq W^* \]
   satisfying $|T|\le 9$, such that:
    \begin{enumerate}
        \item for each path $P\in\mathcal{R}$ that contains every node in $T$, for all $i$ the node $a_i$ appears before $b_i$ on $P$, and
        \item for all $w\in W^*$, $w$ is on a shortest path from $a_i$ to $b_i$ for exactly one $i$. 
    \end{enumerate}
    Then, given these conditions, we can perform a sequence of subpath swaps so that in the resulting path collection $\mathcal{S'}$, some path $P'$ passes through all nodes in $W^*$, and all paths in $\mathcal{R}$ that do not contain any nodes in $W^*\setminus T$ appear in $\mathcal{R}'$.
\end{lemma}

\begin{proof}
Partition the set $W^*\setminus T$ into $W_1,\dots,W_r$ where $w\in W_i$ if $a_i\rightarrow w\rightarrow b_i$ is a shortest-path ordering. For all $u\in U$, if $u$ is on a shortest path from $a_i$ to $b_i$ for some $i$, then $u$ is also included in $W_i$. By condition 2 of the lemma statement, every node $w\in W^*\setminus T$ is in exactly one set $W_i$, and every node $u\in U$ is in at most one set $W_i$. Now, label all of the nodes in $\cup_i (W_i\cup \set{a_i,b_i})$ by $v_1,v_2,v_3\dots$ according to the following order: $a_1,W_1,b_1,a_2,W_2,b_2,\dots a_r,W_r,b_r$, where the nodes within each $W_i$ are ordered by distance from $a_i$. (For the sake of this definition we break ties arbitrarily, but in reality, there cannot be any ties.)

For ease of notation, we consider $\mathcal{R}$ to be changing over time via subpath swaps, and we will let $\mathcal{R}$ denote the current value of $\mathcal{R}$. 

We proceed by induction. Suppose $\mathcal{R}$ contains a path $P$ containing $\{v_1,\dots,v_\ell\}\cup T$ for some $\ell$. Further suppose assumption 1 of the lemma statement: for each path in $\mathcal{R}$ that contains every node in $T$, for all $i$ the node $a_i$ appears before $b_i$. We will show that we can perform subpath swaps so that the resulting path collection $\mathcal{R}$ has a path that contains $\{v_1,\dots,v_{\ell+1}\}\cup T$, and assumption 1 of the lemma statement still holds. 

If $v_{\ell+1}\in T$, then we are already done, so suppose $v_{\ell+1}\not\in T$. Let $j$ be such that $v_{\ell+1}\in W_j$. By assumption 1 of the lemma statement, $a_i$ appears before $b_i$ on $P$, for all $i$. By \cref{lem:sp}, for all $i$, every node $s\in W_i\cap P$ appears between $a_i$ and $b_i$ on $P$. 

We also claim that every node in $\{v_1,\dots,v_{\ell-1}\}\cap W_j$ appears between $a_j$ and $v_{\ell}$ on $P$. From the previous paragraph, we know that every node in $\{v_1,\dots,v_{\ell}\}\cap W_j$ appears between $a_j$ and $b_j$ on $P$. Along the subpath of $P$ from $a_j$ to $b_j$, the nodes of $\{v_1,\dots,v_{\ell}\}\cap W_j$ appear in that order, simply because the $v_i$'s were defined to be ordered by distance from $a_j$.  Thus, every node in $\{v_1,\dots,v_{\ell-1}\}\cap W_j$ appears between $a_j$ and $v_{\ell}$ on $P$.

Consider the set $\{v_1,v_\ell,v_{\ell+1}\}\cup T$, which is of size at most $11$ since $|T|\leq 9$ and $v_1\in T$. 
By property $(\dagger)$ there exists a path $P'\in \mathcal{R}$ that contains every node in $\{v_1,v_\ell,v_{\ell+1}\}\cup T$. 
By assumption 1 of the lemma statement, $a_i$ appears before $b_i$ on $P'$ for all $i$. 

Now we can define our sequence of subpath swaps. For all $i<j$, we swap the subpath $P[a_i,b_i]$ with $P'[a_i,b_i]$. We also swap the subpath $P[a_j,v_\ell]$ with $P'[a_j,v_\ell]$. From above we know that before these swaps, for all $i<j$ every node $w\in W_i$ appears between $a_i$ and $b_i$ on $P$. Thus, after these swaps, for all $i<j$ every node $w\in W_i$ appears between $a_i$ and $b_i$ on $P'$. Similarly, before these swaps, every node in $\{v_1,\dots,v_{\ell-1}\}\cap W_j$ appears between $a_j$ and $v_{\ell}$ on $P$, so after these swaps every node in $\{v_1,\dots,v_{\ell-1}\}\cap W_j$ appears between $a_j$ and $v_{\ell}$ on $P'$. 

Now, we we consider the set $U$. Before these swaps, $P'$ contains all of $U$, and we claim that after the swap $P'$ still contains all of $U$. Fix $u\in U$. If $u\in \{v_1,\dots,v_{\ell-1}\}\cap W_j$, then $P'$ contains $u$ after the swaps, from the previous paragraph. Otherwise, before the swaps, $u$ is not on a path segment of $P'$ that got swapped, so $u$ remains on $P'$. Thus, we have shown that after the swaps, $P'$ contains every node in $\{v_1,\dots,v_{\ell+1}\}\cup T$ as desired.

For the induction, we need to verify that after the swap, assumption 1 still holds. This is simply because we only changed paths $P$ and $P'$, and after the swap, both of them still have $a_i$ appearing before $b_i$ for all $i$.

Finally, we need to verify that all paths in $\mathcal{R}$ that do not contain any nodes in $W^*\setminus T$ are also contained in $\mathcal{S'}$.
This is true simply because the path $P'$ is the only other path besides $P$ that is affected by the swaps and by definition $P'$ contains nodes in $W^*\setminus T$.
\end{proof}

\subsection{Proof of \texorpdfstring{\cref{thm:dirgen}}{Lemma 2.8}}\label{subsec:lem}
We now use the lemmas proved in the previous sections to complete the proof of \cref{thm:dirgen}.

  \paragraph*{Proof Structure}
  
  The structure of the proof is as follows. Let $P$ be the path of $\mathcal{R}$ passing through the most nodes in $W$ (with ties are broken arbitrarily). If $P$ passes through every node in $W$, then we are done.

     Otherwise, as stated in the technical overview, we consider two cases: in case 1 we perform subpath swaps to construct a path $P'$ so that $P$ and $P'$ together satisfy condition 2 of \cref{thm:dirgen}, in which case we are done, and in case 2 we augment one path.

   More, specifically, for case 2, let $W'=P\cap W$. We construct a sequence of subpath swaps that yields a path $P'$ passing through $W'$ and at least one node in $W\setminus W'$.

      After executing case 2, if we are not done we redefine $P'$ to be the path $P$ and repeatedly apply the two cases until we are done.
    Since each iteration of case 2 causes the number of nodes of $W$ on the path $P$ to grow, repeated application of these cases completes the proof. 
    In the full proof below, case 2 is actually broken into 4 cases labeled cases 2-5.

    To begin the description of the cases, fix an iteration, and fix the current value of $\mathcal{R}$. For ease of notation, we consider $\mathcal{R}$ to be changing over time via subpath swaps, and we will always let $\mathcal{R}$ denote the current value of $\mathcal{R}$. Fix the current path $P$. Let $W'=P\cap W$. Let $a$ and $b$ be the first and last nodes on $P$ from $W$. 

   For all of the cases, it will help to have the following definition:

\begin{definition}[Trapping nodes]
   For any node $v\in W$, we say that the \emph{trapping nodes} of $v$, denoted $\text{trap}(v)$ is the pair $\set{\text{trap}_1(v),\text{trap}_2(v)}$ of nodes in $W'$ such that $\dist(a,\text{trap}_1(v))$ is maximized and $\dist(a,\text{trap}_2(v))$ is minimized
   under the condition that $\dist(a,\text{trap}_1(v))\leq\dist(a,v)\leq\dist(a,\text{trap}_2(v))$.
\end{definition}

    \paragraph*{Case 1: Create Two Paths} Our assumption in case 1 is that for \emph{all} $v\in W\setminus W'$, \emph{all} paths in $\mathcal{R}$ that contain $\set{a,b,v}\cup \text{trap}(v)$ have $b\rightarrow v \rightarrow a$ in that order (not necessarily consecutively). Then, letting $v$ be an arbitrary node in $W\setminus W'$, we can apply \cref{lem:reroute} (the Critical Node Lemma) with $W^* =W\setminus W'$, and $T=\set{a,b,v}\cup \text{trap}(v)$, where $a_1=a$, $b_1=b$, and $U=\set{v}\cup \text{trap}(v)$. After the subpath swaps from \cref{lem:reroute}, $\mathcal{R}$ contains a path $P'$ containing every node in $W\setminus W'$. Furthermore,  \cref{lem:reroute} guarantees that $P$ still appears in $\mathcal{R}$. Thus, $\mathcal{R}$ contains two paths $P$ and $P'$ whose union contains every node in $W$.

        To show that $P$ and $P'$ satisfy condition 2 of \cref{thm:dirgen}, we need to show that $b$ is the first node on $P'$ from $W$ and $a$ is the last. This is true simply because of the assumption that for all $v\in W\setminus W'$, all paths in $\mathcal{R}$ that contain $\set{a,b,v}\cup \text{trap}(v)$ have $b\rightarrow v \rightarrow a$ in that order.

        \paragraph*{Case 2: Simple Insertion} Our assumption in case 2 is that there exists a node $v\in W\setminus W'$ such that there exists a path $P'\in\mathcal{R}$ containing any one of the three following orderings of nodes (not necessarily consecutively):
    \begin{itemize}
        \item $a\rightarrow b \rightarrow v$, or
        \item $v \rightarrow a \rightarrow b$, or
        \item $\text{trap}_1(v)\rightarrow v\rightarrow \text{trap}_2(v)$.
    \end{itemize}

    For the first two orderings, we can simply swap the subpath $P[a,b]$ with $P'[a,b]$.
    As a result of this, the path $P'$ contains $W'\cup\set{v}$, so we are done. 
    For the third ordering, we can simply swap $P[\text{trap}_1(v),\text{trap}_2(v)]$ with $P'[\text{trap}_1(v),\text{trap}_2(v)]$. As a result, $P$ contains $W'\cup\set{v}$, so we are done.

    For the sake of analyzing future cases, we claim that we are in case 2 in the following situations: 
    \begin{enumerate}
        \item there exists a node $v\in W\setminus W'$ such that there exists a path $P'\in \mathcal{R}$ containing $\set{a,b,v}\cup \text{trap}(v)$ that has $a$ before $b$, or
        \item the subpath $P[a,b]$ is a trivial path with respect to some $w\in W'$. In this case, by the definition of a trivial path there exists a path $P'\in\mathcal{R}$ containing $\set{a,b,v,w}\cup \text{trap}(v)$ that has $a$ before $b$ by property $(\dagger)$.
    \end{enumerate} 

   For either situation, fix such a path $P'$. If $P'$ has $v$ before $a$ or after $b$, then we fall into case 2 by definition.
Otherwise, $P'$ has $v$ between $a$ and $b$. By \cref{lem:sp}, $P'$ also has both nodes in $\text{trap}(v)$ between $a$ and $b$.
    Because $\set{v}\cup \text{trap}(v)$ all appear between $a$ and $b$ on $P'$, we know that since $\dist(a,\text{trap}_1(v))\leq\dist(a,v)\leq\dist(a,\text{trap}_2(v))$ they appear in the order $\text{trap}_1(v)\rightarrow v\rightarrow \text{trap}_2(v)$ on $P'$. Thus, we fall into case 2.

        Therefore, for the remaining cases, we know that the subpath $P[a,b]$ is not a trivial path with respect to any node in $W'$. Thus, we have two cases to consider: $P[a,b]$ is a reversing path with respect to $W'$, or  $P[a,b]$ is non-reversing with respect to $W'$.

    \paragraph*{Case 3: Reversing Path}
    
    Suppose $P[a,b]$ is a reversing path with respect to $W'$. Let $v$ be an arbitrary node in $W\setminus W'$. Our goal is to perform subpath swaps so that as a result, $\mathcal{R}$ has a path that contains every node in $W'\cup \set{v}$.

    We apply \cref{lem:combo} (the Reversing/Non-Reversing Path Lemma) to $P[a,b]$, where $W'$ in \cref{lem:combo} is the same as our set $W'$, to get
   a decomposition of $P[a,b]$ into three segments $P_1, P_2, P_3$.

    Let $w$ be the last node on $P_1$ from $W'$, 
    and let $z$ be the first node on $P_3$ from $W'$. That is,
 $P[a,w]\subseteq P_1$, 
 and $P[z, b]\subseteq P_3$.

      Let $p$ and $q$ be the nodes in $P_2\cap W$ which are the closest and furthest from $b$ in terms of distance respectively (such nodes exist because $P_2$ is non-empty by the definition of a reversing path).
    Define the set $K = \set{a, b, w, z, p, q, v}\cup \text{trap}(v)$. $K$ will be our set of \emph{critical nodes}. By property $(\dagger)$, there is a path $Q\in\mathcal{R}$ containing every node in $K$.

    Next, we will analyze the order in which $Q$ traverses the nodes of $K$. After establishing this structure, we will perform subpath swaps by applying \cref{lem:reroute} (the Critical Node Lemma).

   Because we are not in {\bf case 2}, we know that \emph{all} paths in $\mathcal{R}$ that contain $\set{a,b,v}\cup \text{trap}(v)$ have $b$ before $a$, and in particular $Q$. Because $p$ and $q$ are in $P_2$, by \cref{lem:combo} $p$ and $q$ are both on a shortest path from $b$ to $a$. Thus, by \cref{lem:sp}, $p$ and $q$ both appear between $b$ and $a$ on $Q$.  Furthermore, by considering the properties of $P_1$ and $P_3$ from the statement of \Cref{lem:combo}, we know that $z$ occurs before $b$ on $Q$, and $w$ occurs after $a$ on $Q$.  By our choice of $p$ and $q$ as the closest and furthest nodes from $b$ respectively, we know that $p$ appears before $q$ on $Q$.

       In summary, $Q$ is such that $z$ appears before $b$, $p$ appears before $q$, and $a$ appears before $w$. With this structure, we are ready to apply \cref{lem:reroute} with the following choice of parameters: $T$ is set to $K$, $W^*$  is set to $W'\cup \set{v}$, $(a_1,b_1),(a_2,b_2),(a_3,b_3)$ are set to $(z,b),(p,q),(a,w)$, and $U$ is set to $\set{v}\cup \text{trap}(v)$. We will now verify that the assumptions of \cref{lem:reroute} indeed hold. For the first assumption, $Q$ was chosen to be an arbitrary path in $\mathcal{R}$ containing every node in $K$, so the properties we proved for $Q$ are indeed true for every such path. For the second assumption, by definition, every node in $P_1\cap W$ falls on a shortest path from $a$ to $w$, every node in $P_2\cap W$ falls on a shortest path from $b$ to $a$, and every node in $P_3\cap W$ falls on a shortest path from $z$ to $b$. No node falls into multiple path segments $P_i$, so each node in $W'$ is indeed on a shortest path from $a_i$ to $b_i$ for exactly one $i$.  The consequence of \cref{lem:reroute} is that $\mathcal{R}$ has a path containing every node in $W'\cup\set{v}$, as desired.

        The remaining cases consider the situation where $P[a,b]$ is a non-reversing path with respect to $W'$.

        \paragraph*{Non-Reversing Path Lemma}
When $P[a,b]$ is a non-reversing path with respect to $W'$, we will split into two cases. Before presenting either of the cases, we prove a lemma that will be useful for both cases.

\begin{lemma}\label{lem:cycle}
Suppose $u,w\in W'$ and $u$ is before $w$ on $P$. Let $u'$ be any node in $W'$ that appears before $u$ on $P$, and let $w'$ be any node in $W'$ that appears after $w$ on $P$. Then, any path $Q\in \mathcal{R}$ containing $\set{u,w,u',w',a,b}$ that has $w$ before $u$, has the order $w\to w' \to b \to a\to u' \to u$.
\end{lemma}

\begin{proof} Let $Q\in \mathcal{R}$ be a path containing $\set{u,w,u',w',a,b}$ that has $w$ before $u$.
    Node $a$ must occur after $w$ on $Q$, as otherwise we would have the shortest-path ordering $a\rightarrow w \rightarrow u$, which by \cref{lem:sp} contradicts the shortest-path ordering $a\rightarrow u \rightarrow w$, which we have by assumption.
    Symmetric reasoning shows that node $b$ must occur before $u$ on $Q$.
    Finally, $b$ must appear before $a$ on $Q$, because otherwise we would have the shortest-path ordering $w\rightarrow a \rightarrow b \rightarrow u$, which by \cref{lem:sp} contradicts the shortest-path ordering $a \rightarrow u\rightarrow w \rightarrow b$, which we have by assumption.
    
        So far we have simply used distance considerations to restrict the order these nodes appear on $Q$. 
    We now use the fact that $P[a,b]$ is non-reversing (with respect to $W'$) to pin down the precise order that $Q$ passes through these nodes.
    
    We claim that $b$ appears after $w$ on $Q$.
    Indeed, suppose to the contrary that $b$ appeared before $w$.
    Then since $w$ appears before $u$, which in turn appears before $a$, on $Q$, we would have the shortest-path ordering $b\rightarrow w\rightarrow a$.
    However, this contradicts the fact that $P[a,b]$ is non-reversing with respect to $W'$.  
    Symmetric reasoning shows that $a$ appears after $u$ on $Q$.
    
    Combining all these observations together, we see that $Q$ traverses these nodes in the order
    $w\to b\to a\to u$.
    
    Since $w'$ falls on a shortest path from $w$ to $b$, by \cref{lem:sp} $w'$ falls between $w$ and $b$ on $Q$. Similarly, since $u'$ is on a shortest path from $a$ to $u$, by \cref{lem:sp} $u'$ falls between $a$ and $u$ on $Q$. Thus, $Q$ has the order $w\to w' \to b \to a\to u' \to u$, as desired.
\end{proof}

   \paragraph*{Case 4: Non-Reversing Path Easy Case}
   Suppose $P[a,b]$ is a non-reversing path with respect to $W'$. Our further assumption in case 4 is that there exists a node $v\in W\setminus W'$ with
      the following property: there exist nodes $u,w\in W'$ where $u$ is before $w$ on $P$ and no node of $W'$ falls between $u$ and $w$ on $P$, such that every path in $\mathcal{R}$ containing $\set{u,w,v}$, has $w$ before $u$. Fix such a node $v$. Our goal is to perform subpath swaps so that as a result, $\mathcal{R}$ has a path that contains every node in $W'\cup \set{v}$.

  Define the set $K=\set{v, a, b, u, w}$, which will be our set of \emph{critical nodes}. By property $(\dagger)$ there is a path $Q\in \mathcal{R}$ passing through all the nodes in $K$.
      By assumption, $w$ occurs before $u$ on $Q$. By \cref{lem:cycle}, $Q$ traverses these nodes in the order
     $w\to b\to a\to u$.

    Now we are ready to apply \Cref{lem:reroute} (the Critical Node Lemma) with the following choice of parameters: $T$ is set to $K$, $W^*$ is set to $W'\cup \set{v}$, $(a_1,b_1),(a_2,b_2)$ are set to $(w,b),(a,u)$, and $U$ is set to $\set{v}$. We will now verify that the assumptions of \cref{lem:reroute} indeed hold.  For the first assumption, $Q$ was chosen to be an arbitrary path in $\mathcal{R}$ containing every node in $K$, so the properties we proved for $Q$ are indeed true for every such path. For the second assumption, by definition, every node in $W'$ either falls on exactly one of $P[w,b]$ or $P[a,u]$. The consequence of \cref{lem:reroute} is that $\mathcal{R}$ has a path containing every node in $W'\cup\set{v}$, as desired.

       \paragraph*{Case 5: Non-Reversing Path Hard Case}
   In case 5, like in case 4, we assume $P[a,b]$ is a non-reversing path with respect to $W'$.

     Because we are not in {\bf case 2}, we know that \emph{every} node $v\in W\setminus W'$ has the property that \emph{all} paths in $\mathcal{R}$ which contain $\set{a,b,v}\cup \text{trap}(v)$ have $b$ before $a$. 
     Moreover, because we are not in {\bf case 1}, we know there exists a node $v\in W\setminus W'$ such that there is a path in $\mathcal{R}$ containing $\set{a,b,v}\cup \text{trap}(v)$ where $a,b,v$ appear in one of the following two orders: $v\rightarrow b\rightarrow a$ or $b\rightarrow a \rightarrow v$. 
     Fix such a $v$. Our goal is to perform subpath swaps such that afterwards $\mathcal{R}$ has a path that contains every node in $W'\cup \set{v}$.

        Since we are not in {\bf case 4}, for every choice of $u,w\in W'$ where $u$ is before $w$ on $P$ and no node of $W'$ falls between $u$ and $w$ on $P$, there exists a path in $\mathcal{R}$ through $u, w,$ and $v$ in which $u$ appears before $w$. Moreover, because we are not in {\bf case 2}, we can assume that $v$ does not appear between $u$ and $w$ on such a path. It follows that for any choice of such $u$ and $w$, there is a path in $\mathcal{R}$ with the ordering $v\to u \to w$ or $u\to w\to v$.

    Notice that we have reached the same conclusion for the nodes $u,w$ as we have for the nodes $b,a$. Because of this, it will be convenient to imagine the path $P$ as a cycle in the following sense.  Enumerate the nodes in $W'$ in the order that they appear on $P$ as
        $a = x_1, x_2, \dots, x_{|W'|} = b$,
    and interpret the indices modulo $|W'|$, so that $x_i = x_{i\pmod {|W'|}}$ whenever $i > |W'|$. In the language of this notation, we have shown that for all $i$, there is a path in $\mathcal{R}$ with the ordering $v\to x_i \to x_{i+1}$ or $x_i\to x_{i+1}\to v$. 
    We call the former type a \emph{beginning path} and the latter type an \emph{ending path}, referring to the position of $v$ relative to $x_i$ and $x_{i+1}$ on the respective paths.

        In the next step of the argument, we consider for which values of $i$ the nodes $x_i,x_{i+1}$ admit a beginning path, for which values of $i$ they admit an ending path, and for which values of $i$ they admit both.

     Let $s$ be such that $x_{s+1}$ is the node in $W'$ with minimum distance from $v$. 
    Then no path in $\mathcal{R}$ through $x_{s}$, $x_{s+1}$, and $v$ can be a beginning path, because this would mean that $\dist(v,x_{s}) < \dist(v,x_{s+1})$.
        Thus, $\mathcal{R}$ must have an ending path through $x_{s}$, $x_{s+1}$, and $v$.

        Symmetrically, by considering the node in $W'$ with minimum distance to $v$, we conclude that there exists an $i$ such that there is no ending path through $x_i$, $x_{i+1}$

            Let $t$ be the integer greater than $s$ that minimizes $t-s \pmod {|W'|}$ such that there is no ending path through $x_t$, $x_{t+1}$, and $v$ (such a $t$ exists due to the previous paragraph). Thus, $\mathcal{R}$ must have a beginning path through $v$, $x_t$, $x_{t+1}$. 
    Now, define our set $K$ of \emph{critical nodes} as
    \begin{equation}
    \label{eq:K-case5}
    K = \set{a, b, x_s, x_{s+1}, x_t, x_{t+1}, v}\cup \text{trap}(v).
    \end{equation}

        By property $(\dagger)$ there exists a path $Q\in\mathcal{R}$ containing all of the nodes in $K$.
        
            We now proceed with casework on the relative orderings of these nodes on $Q$.

    Suppose that $x_s$ occurs before $x_{t+1}$ on $Q$. We begin by proving the following claim.
    \begin{claim}\label{claim:bet}
        Node $v$ is between $x_s$ and $x_{t+1}$ on $Q$.
    \end{claim}
  \begin{proof}
    Suppose for sake of contradiction that $Q$ passes through the nodes from the claim statement either in the order $v\to x_s\to x_{t+1}$ or in the order $x_s\to x_{t+1}\to v$.
    
        We claim that the nodes $x_{s+1}$ and $x_t$ must appear (in that order) between $x_s$ and $x_{t+1}$ on $Q$.
    To see this, consider the following cases: If $x_s$ comes before $x_{t+1}$ on $P$ then the order of these nodes on $P$ is $x_s\to x_{s+1} \to x_t \to x_{t+1}$, in which case apply \cref{lem:sp}. On the other hand, if $x_{t+1}$ comes before $x_s$ on $P$ then there are two possible orderings of these nodes on $P$: $x_t \to x_{t+1} \to x_s\to x_{s+1}$, in which case apply \cref{lem:cycle}, or $x_{t+1} \to x_s\to x_{s+1} \to x_t$, in which case apply \cref{lem:sp}.

         Thus, in the first case where $Q$ has the ordering $v\to x_s\to x_{t+1}$, $Q$ must have the ordering $v\to x_s\to x_{s+1} \to x_{t+1}$, so there is a path in $\mathcal{R}$ with the ordering $v\to x_s\to x_{s+1}$, which contradicts our choice of $s$. 
    Similarly, in the second case where $Q$ has $x_s\to x_{t+1}\to v$, $Q$ must have $x_s\to x_t \to x_{t+1}\to v$, so there is a path in $\mathcal{R}$ with the ordering $x_t\to x_{t+1}\to v$, which contradicts our choice of $t$. 
    \end{proof}

    Next, we prove a claim that we can perform subpath swaps to obtain a path in $\mathcal{R}$ through $x_s\to x_{s+1} \to \dots \to x_{t} \to v$, where the dots indicate all $x_i$ between $x_{s+1}$ and $x_{t}$: 

    \begin{claim}\label{claim:induct}
        For any index $j$ such that $s< j\le t$, we can perform subpath swaps so that as a result $\mathcal{R}$ has a path with the ordering $x_s\to x_{s+1} \to \dots \to x_{j} \to v$ (where the dots indicate all $x_i$ between $x_{s+1}$ and $x_{j}$).
    \end{claim}

    \begin{proof}
    We proceed by induction on $j$.
    
        The base case of $j=s+1$ holds simply because $\mathcal{R}$ has an ending path through $x_s, x_{s+1}$, and $v$ (in that order) by choice of $s$.
    
    For the inductive step, suppose that for some index $j$ such that $s < j < t$, $\mathcal{R}$ has a path $P'$ with the ordering $x_s\to x_{s+1}\to \dots\to x_j\to v$.
    By our choice of $t$, we know that $\mathcal{R}$ has an ending path $P''$ passing through $x_j, x_{j+1}$, and $v$ (in that order).
    Now, we simply swap $P'[x_j,v]$ with $P''[x_j,v]$. As a result, $P'$ has the ordering $x_s\to x_{s+1}\to \dots\to x_{j+1}\to v$.

    This completes the induction.
      \end{proof}

    Let $P'$ be the path from \cref{claim:induct} with $j=t$. Swap $P'[x_{s},v]$ with $Q[x_{s}, v]$; we can do so because $x_s$ falls before $v$ on $Q$ by \cref{claim:bet}. 
      As a result, $Q$ contains the ordering 
      \[x_s\to x_{s+1}\to \dots\to x_{t}\to v\to x_{t+1}.\]

Thus, $x_{t}\to v\to x_{t+1}$ is a shortest-path ordering. However, from the definition of $t$, $\mathcal{R}$ has a beginning path through $v$, $x_t$, $x_{t+1}$, so $v\to x_t\to x_{t+1}$ is also a shortest-path ordering. These two orderings contradict \cref{lem:sp}. Thus, case 5(a) is impossible in the first place.

   \subparagraph*{Case 5(b): $x_{t+1}$ before $x_s$}$ $\\
    
    Since {\bf case 5(a)} cannot happen, we know that $x_{t+1}$ appears before $x_s$ on $Q$.
    
    First, we determine the relative ordering of $x_{s+1}$ and $x_t$ on $Q$ with the following claim:

    \begin{claim}
        \label{claim:case5b-order}
      Node $x_{s+1}$ appears before $x_t$ on $Q$.
    \end{claim}

        \begin{proof}
       Suppose for the sake of contradiction that $x_t$ appears before $x_{s+1}$ on $Q$. 

              By \Cref{lem:cycle} we know that the ordering of $x_s, x_{s+1}, x_t, x_{t+1}$ on any path in $\mathcal{R}$ containing \[\set{x_s, x_{s+1}, x_t, x_{t+1},a,b}\] must be equivalent to
            \[x_s\to x_{s+1}\to x_t\to x_{t+1}\]
        up to cyclic reordering.

       Applying the assumption of case 5(b) that $x_{t+1}$ appears before $x_s$ on $Q$, and the assumption (for contradiction) that $x_t$ appears before $x_{s+1}$ on $Q$, we conclude that $Q$ passes through these nodes in the order
            \[x_{t}\to x_{t+1}\to x_s \to x_{s+1}.\]

        Now, where can $v$ be located on $Q$, relative to the above nodes?
      By our choice of $s$, we know that $v$ cannot occur before $x_s$, because if it did then $Q$ would go through $v\to x_s\to x_{s+1}$ and thus be a beginning path.
        Thus, $v$ must occur after $x_s$. Thus, $v$ occurs after $x_{t+1}$.
        But then $Q$ goes through $x_t\to x_{t+1}\to v$ as an ending path.
        This contradicts our choice of $t$ and completes the proof.
            \end{proof}

    With \Cref{claim:case5b-order} established, we know that $x_{s+1}$ occurs before $x_t$, and $x_{t+1}$ occurs before $x_s$ on $Q$. Additionally, recall from the first paragraph of case 5 that because we are not in {\bf case 2}, $b$ falls before $a$ on $Q$.

        Consider the position of $b$ on the cyclic order given by the $x_i$'s. There are 3 cases:
    \begin{enumerate}
        \item $b=x_s$ (and thus $a=x_{s+1}$) or $b=x_t$ (and thus $a=x_{t+1}$),
        \item $b$ is between $x_{s+1}$ and $x_{t-1}$ (inclusive), or
        \item $b$ is between $x_{t+1}$ and $x_{s-1}$ (inclusive).
    \end{enumerate}

    Each of these 3 cases are similar, and for each we will apply \Cref{lem:reroute} (the Critical Node Lemma). For all 3 instantiations of \cref{lem:reroute}, $T$ will be set to $K$ where $K$ is defined in \cref{eq:K-case5}, $W^*$ will be set to $W'\cup \set{v}$, and $U$ will be set to $\text{trap}(v)\cup \set{v}$. The pairs $(a_i,b_i)$ will be different for each of the 3 cases.  The consequence of \cref{lem:reroute} for all 3 cases is that $\mathcal{R}$ has a path containing every node in $W'\cup\set{v}$, as desired. Thus, it remains to define the pairs $(a_i,b_i)$ for each of the 3 cases and verify that the assumptions of \cref{lem:reroute} hold for each case.

    In case 1, $(a_1, b_1)$ and $(a_2, b_2)$ are defined as $(x_{t+1},x_s)$ and $(x_{s+1},x_t)$. We will now verify that the assumptions of \cref{lem:reroute} hold. For the first assumption, $Q$ was chosen to be an arbitrary path in $\mathcal{R}$ containing every node in $K$, so the properties we proved for $Q$ are indeed true for every such path. For the second assumption, $P[x_{t+1},x_s]\cup P[x_{s+1},x_t]$ by definition contains every node in $W'$, and these subpaths are non-overlapping.

    In case 2, $(a_1, b_1)$, $(a_2, b_2)$, and $(a_3,b_3)$ are defined as $(a,x_t)$, $(x_{t+1},x_s)$, and $(x_{s+1},b)$ (where the last pair is not included if $b=x_{s+1}$).  We will now verify that the assumptions of \cref{lem:reroute} hold.

    For the first assumption, $a$ appears before $x_t$ on $Q$ because otherwise $Q$ would have the ordering $x_{s+1}\to x_t\to a$, where $b$ is between $x_t$ and $a$ because $b$ cannot fall before $x_{s+1}$ or $x_t$ due to the orderings imposed by \cref{lem:combo} for non-reversing paths. 
    This contradicts the ordering on $P$ of $a\to x_t \to x_{s+1}\to b$ by \cref{lem:sp}. Similarly, $x_{s+1}$ appears before $b$ for a symmetric reason: otherwise $Q$ would have the ordering $b\to a \to x_{s+1} \to x_t$ which again contradicts the ordering of these nodes on $P$. $Q$ was chosen to be an arbitrary path in $\mathcal{R}$ containing every node in $K$, so the properties we proved for $Q$ are true for every such path. 
    
    For the second assumption, $P[a,x_t]\cup P[x_{t+1},x_s]\cup P[x_{s+1},b]$ by definition contains every node in the set $W'$, and these subpaths are non-overlapping. 
    
    Case 3 is precisely the same as case 2 but with $s$ and $t$ switched.
    
   This completes case 5, which completes the proof of \cref{thm:dirgen}.

\section*{Acknowledgements}

We thank Virginia Vassilevska Williams for helpful discussions related to this work.

\bibliographystyle{alpha}
\bibliography{main.bib}

\newcommand{\etalchar}[1]{$^{#1}$}
\begin{thebibliography}{AKMR19}

\bibitem[AKMR19]{DBLP:journals/ipl/AmiriKMR19}
Saeed~Akhoondian Amiri, Stephan Kreutzer, D{\'{a}}niel Marx, and Roman
  Rabinovich.
\newblock Routing with congestion in acyclic digraphs.
\newblock {\em Inf. Process. Lett.}, 151, 2019.

\bibitem[And10]{DBLP:conf/focs/Andrews10}
Matthew Andrews.
\newblock Approximation algorithms for the edge-disjoint paths problem via
  raecke decompositions.
\newblock In {\em 51th Annual {IEEE} Symposium on Foundations of Computer
  Science, {FOCS} 2010, October 23-26, 2010, Las Vegas, Nevada, {USA}}, pages
  277--286. {IEEE} Computer Society, 2010.

\bibitem[AR06]{DBLP:journals/algorithmica/AzarR06}
Yossi Azar and Oded Regev.
\newblock Combinatorial algorithms for the unsplittable flow problem.
\newblock {\em Algorithmica}, 44(1):49--66, 2006.

\bibitem[AW20]{SDP-congestion-DAG}
Saeed~Akhoondian Amiri and Julian Wargalla.
\newblock Disjoint shortest paths with congestion on dags.
\newblock {\em CoRR}, abs/2008.08368, 2020.

\bibitem[AWW16]{DBLP:conf/soda/AbboudWW16}
Amir Abboud, Virginia~Vassilevska Williams, and Joshua~R. Wang.
\newblock Approximation and fixed parameter subquadratic algorithms for radius
  and diameter in sparse graphs.
\newblock In Robert Krauthgamer, editor, {\em Proceedings of the Twenty-Seventh
  Annual {ACM-SIAM} Symposium on Discrete Algorithms, {SODA} 2016, Arlington,
  VA, USA, January 10-12, 2016}, pages 377--391. {SIAM}, 2016.

\bibitem[AZ07]{DBLP:journals/siamcomp/AndrewsZ07}
Matthew Andrews and Lisa Zhang.
\newblock Hardness of the undirected congestion minimization problem.
\newblock {\em {SIAM} J. Comput.}, 37(1):112--131, 2007.

\bibitem[BH14]{bjorklund2014shortest}
Andreas Bj{\"o}rklund and Thore Husfeldt.
\newblock Shortest two disjoint paths in polynomial time.
\newblock In {\em International Colloquium on Automata, Languages, and
  Programming}, pages 211--222. Springer, 2014.

\bibitem[BH23]{BodwinWinning}
Greg Bodwin and Gary Hoppenworth.
\newblock Folklore sampling is optimal for exact hopsets: Confirming the
  $\sqrt{n}$ barrier, 2023.

\bibitem[BK17]{DBLP:conf/esa/Berczi017}
Krist{\'{o}}f B{\'{e}}rczi and Yusuke Kobayashi.
\newblock The directed disjoint shortest paths problem.
\newblock In Kirk Pruhs and Christian Sohler, editors, {\em 25th Annual
  European Symposium on Algorithms, {ESA} 2017, September 4-6, 2017, Vienna,
  Austria}, volume~87 of {\em LIPIcs}, pages 13:1--13:13. Schloss Dagstuhl -
  Leibniz-Zentrum f{\"{u}}r Informatik, 2017.

\bibitem[BNRZ21]{SDP-undirected-geometric}
Matthias Bentert, Andr\'{e} Nichterlein, Malte Renken, and Philipp Zschoche.
\newblock {Using a Geometric Lens to Find k Disjoint Shortest Paths}.
\newblock In Nikhil Bansal, Emanuela Merelli, and James Worrell, editors, {\em
  48th International Colloquium on Automata, Languages, and Programming (ICALP
  2021)}, volume 198 of {\em Leibniz International Proceedings in Informatics
  (LIPIcs)}, pages 26:1--26:14, Dagstuhl, Germany, 2021. Schloss Dagstuhl --
  Leibniz-Zentrum f{\"u}r Informatik.

\bibitem[BNZ15]{DBLP:conf/esa/BorradaileNZ15}
Glencora Borradaile, Amir Nayyeri, and Farzad Zafarani.
\newblock Towards single face shortest vertex-disjoint paths in undirected
  planar graphs.
\newblock In Nikhil Bansal and Irene Finocchi, editors, {\em Algorithms - {ESA}
  2015 - 23rd Annual European Symposium, Patras, Greece, September 14-16, 2015,
  Proceedings}, volume 9294 of {\em Lecture Notes in Computer Science}, pages
  227--238. Springer, 2015.

\bibitem[Bod17]{DBLP:conf/soda/Bodwin17}
Greg Bodwin.
\newblock Linear size distance preservers.
\newblock In Philip~N. Klein, editor, {\em Proceedings of the Twenty-Eighth
  Annual {ACM-SIAM} Symposium on Discrete Algorithms, {SODA} 2017, Barcelona,
  Spain, Hotel Porta Fira, January 16-19}, pages 600--615. {SIAM}, 2017.

\bibitem[Bod19]{DBLP:conf/soda/Bodwin19}
Greg Bodwin.
\newblock On the structure of unique shortest paths in graphs.
\newblock In Timothy~M. Chan, editor, {\em Proceedings of the Thirtieth Annual
  {ACM-SIAM} Symposium on Discrete Algorithms, {SODA} 2019, San Diego,
  California, USA, January 6-9, 2019}, pages 2071--2089. {SIAM}, 2019.

\bibitem[BS00]{DBLP:journals/mor/BavejaS00}
Alok Baveja and Aravind Srinivasan.
\newblock Approximation algorithms for disjoint paths and related routing and
  packing problems.
\newblock {\em Math. Oper. Res.}, 25(2):255--280, 2000.

\bibitem[BW22]{nicole-aaron-amazing-work}
Aaron Bernstein and Nicole Wein.
\newblock Closing the gap between directed hopsets and shortcut sets, 2022.

\bibitem[CE13]{DBLP:conf/soda/ChekuriE13}
Chandra Chekuri and Alina Ene.
\newblock Poly-logarithmic approximation for maximum node disjoint paths with
  constant congestion.
\newblock In Sanjeev Khanna, editor, {\em Proceedings of the Twenty-Fourth
  Annual {ACM-SIAM} Symposium on Discrete Algorithms, {SODA} 2013, New Orleans,
  Louisiana, USA, January 6-8, 2013}, pages 326--341. {SIAM}, 2013.

\bibitem[CGR16]{DBLP:conf/soda/CairoGR16}
Massimo Cairo, Roberto Grossi, and Romeo Rizzi.
\newblock New bounds for approximating extremal distances in undirected graphs.
\newblock In Robert Krauthgamer, editor, {\em Proceedings of the Twenty-Seventh
  Annual {ACM-SIAM} Symposium on Discrete Algorithms, {SODA} 2016, Arlington,
  VA, USA, January 10-12, 2016}, pages 363--376. {SIAM}, 2016.

\bibitem[Chu16]{DBLP:journals/siamcomp/Chuzhoy16}
Julia Chuzhoy.
\newblock Routing in undirected graphs with constant congestion.
\newblock {\em {SIAM} J. Comput.}, 45(4):1490--1532, 2016.

\bibitem[CKS09]{DBLP:journals/siamcomp/ChekuriKS09}
Chandra Chekuri, Sanjeev Khanna, and F.~Bruce Shepherd.
\newblock Edge-disjoint paths in planar graphs with constant congestion.
\newblock {\em {SIAM} J. Comput.}, 39(1):281--301, 2009.

\bibitem[CL22a]{Cizma2022}
Daniel Cizma and Nati Linial.
\newblock Geodesic geometry on graphs.
\newblock {\em Discrete \& Computational Geometry}, 68(1):298--347, January
  2022.

\bibitem[CL22b]{Cizma2023}
Daniel Cizma and Nati Linial.
\newblock Irreducible nonmetrizable path systems in graphs.
\newblock {\em Journal of Graph Theory}, 102(1):5--14, June 2022.

\bibitem[CMPP13]{DBLP:conf/focs/CyganMPP13}
Marek Cygan, D{\'{a}}niel Marx, Marcin Pilipczuk, and Michal Pilipczuk.
\newblock The planar directed k-vertex-disjoint paths problem is
  fixed-parameter tractable.
\newblock In {\em 54th Annual {IEEE} Symposium on Foundations of Computer
  Science, {FOCS} 2013, 26-29 October, 2013, Berkeley, CA, {USA}}, pages
  197--206. {IEEE} Computer Society, 2013.

\bibitem[CSS15]{chudnovsky2015disjoint}
Maria Chudnovsky, Alex Scott, and Paul Seymour.
\newblock Disjoint paths in tournaments.
\newblock {\em Advances in Mathematics}, 270:582--597, 2015.

\bibitem[CW99]{DBLP:conf/soda/CowenW99}
Lenore Cowen and Christopher~G. Wagner.
\newblock Compact roundtrip routing for digraphs.
\newblock In Robert~Endre Tarjan and Tandy~J. Warnow, editors, {\em Proceedings
  of the Tenth Annual {ACM-SIAM} Symposium on Discrete Algorithms, 17-19
  January 1999, Baltimore, Maryland, {USA}}, pages 885--886. {ACM/SIAM}, 1999.

\bibitem[CZ22]{DBLP:conf/soda/ChechikZ22}
Shiri Chechik and Tianyi Zhang.
\newblock Nearly 2-approximate distance oracles in subquadratic time.
\newblock In Joseph~(Seffi) Naor and Niv Buchbinder, editors, {\em Proceedings
  of the 2022 {ACM-SIAM} Symposium on Discrete Algorithms, {SODA} 2022, Virtual
  Conference / Alexandria, VA, USA, January 9 - 12, 2022}, pages 551--580.
  {SIAM}, 2022.

\bibitem[DIKM18]{DBLP:conf/fsttcs/DattaIK018}
Samir Datta, Siddharth Iyer, Raghav Kulkarni, and Anish Mukherjee.
\newblock Shortest k-disjoint paths via determinants.
\newblock In Sumit Ganguly and Paritosh~K. Pandya, editors, {\em 38th {IARCS}
  Annual Conference on Foundations of Software Technology and Theoretical
  Computer Science, {FSTTCS} 2018, December 11-13, 2018, Ahmedabad, India},
  volume 122 of {\em LIPIcs}, pages 19:1--19:21. Schloss Dagstuhl -
  Leibniz-Zentrum f{\"{u}}r Informatik, 2018.

\bibitem[dVS11]{DBLP:journals/talg/VerdiereS11}
{\'{E}}ric~Colin de~Verdi{\`{e}}re and Alexander Schrijver.
\newblock Shortest vertex-disjoint two-face paths in planar graphs.
\newblock {\em {ACM} Trans. Algorithms}, 7(2):19:1--19:12, 2011.

\bibitem[EIS76]{DBLP:journals/siamcomp/EvenIS76}
Shimon Even, Alon Itai, and Adi Shamir.
\newblock On the complexity of timetable and multicommodity flow problems.
\newblock {\em {SIAM} J. Comput.}, 5(4):691--703, 1976.

\bibitem[EN19]{DBLP:conf/spaa/ElkinN19}
Michael Elkin and Ofer Neiman.
\newblock Linear-size hopsets with small hopbound, and constant-hopbound
  hopsets in {RNC}.
\newblock In Christian Scheideler and Petra Berenbrink, editors, {\em The 31st
  {ACM} on Symposium on Parallelism in Algorithms and Architectures, {SPAA}
  2019, Phoenix, AZ, USA, June 22-24, 2019}, pages 333--341. {ACM}, 2019.

\bibitem[ET98]{eilam1998disjoint}
Tali Eilam-Tzoreff.
\newblock The disjoint shortest paths problem.
\newblock {\em Discrete applied mathematics}, 85(2):113--138, 1998.

\bibitem[FGL{\etalchar{+}}22]{detours}
Fedor~V. Fomin, Petr~A. Golovach, William Lochet, Danil Sagunov, Kirill
  Simonov, and Saket Saurabh.
\newblock {Detours in Directed Graphs}.
\newblock In Petra Berenbrink and Benjamin Monmege, editors, {\em 39th
  International Symposium on Theoretical Aspects of Computer Science (STACS
  2022)}, volume 219 of {\em Leibniz International Proceedings in Informatics
  (LIPIcs)}, pages 29:1--29:16, Dagstuhl, Germany, 2022. Schloss Dagstuhl --
  Leibniz-Zentrum f{\"u}r Informatik.

\bibitem[FHW80]{DBLP:journals/tcs/FortuneHW80}
Steven Fortune, John~E. Hopcroft, and James Wyllie.
\newblock The directed subgraph homeomorphism problem.
\newblock {\em Theor. Comput. Sci.}, 10:111--121, 1980.

\bibitem[HP19]{DBLP:journals/ipl/HuangP19}
Shang{-}En Huang and Seth Pettie.
\newblock Thorup-zwick emulators are universally optimal hopsets.
\newblock {\em Inf. Process. Lett.}, 142:9--13, 2019.

\bibitem[HP21]{DBLP:journals/siamdm/HuangP21}
Shang{-}En Huang and Seth Pettie.
\newblock Lower bounds on sparse spanners, emulators, and diameter-reducing
  shortcuts.
\newblock {\em {SIAM} J. Discret. Math.}, 35(3):2129--2144, 2021.

\bibitem[KK11]{DBLP:conf/stoc/KawarabayashiK11}
Ken{-}ichi Kawarabayashi and Yusuke Kobayashi.
\newblock Breaking o(n\({}^{\mbox{1/2}}\))-approximation algorithms for the
  edge-disjoint paths problem with congestion two.
\newblock In Lance Fortnow and Salil~P. Vadhan, editors, {\em Proceedings of
  the 43rd {ACM} Symposium on Theory of Computing, {STOC} 2011, San Jose, CA,
  USA, 6-8 June 2011}, pages 81--88. {ACM}, 2011.

\bibitem[KKK14]{kawarabayashi2014excluded}
Ken-ichi Kawarabayashi, Yusuke Kobayashi, and Stephan Kreutzer.
\newblock An excluded half-integral grid theorem for digraphs and the directed
  disjoint paths problem.
\newblock In {\em Proceedings of the forty-sixth annual ACM symposium on Theory
  of computing}, pages 70--78, 2014.

\bibitem[KKR12]{DBLP:journals/jct/KawarabayashiKR12}
Ken{-}ichi Kawarabayashi, Yusuke Kobayashi, and Bruce~A. Reed.
\newblock The disjoint paths problem in quadratic time.
\newblock {\em J. Comb. Theory, Ser. {B}}, 102(2):424--435, 2012.

\bibitem[KN04]{DBLP:journals/ipl/KrasikovN04}
Ilia Krasikov and Steven~D. Noble.
\newblock Finding next-to-shortest paths in a graph.
\newblock {\em Inf. Process. Lett.}, 92(3):117--119, 2004.

\bibitem[KS98]{DBLP:conf/ipco/KolliopoulosS98}
Stavros~G. Kolliopoulos and Clifford Stein.
\newblock Approximating disjoint-path problems using greedy algorithms and
  packing integer programs.
\newblock In Robert~E. Bixby, E.~Andrew Boyd, and Roger~Z.
  R{\'{\i}}os{-}Mercado, editors, {\em Integer Programming and Combinatorial
  Optimization, 6th International {IPCO} Conference, Houston, Texas, USA, June
  22-24, 1998, Proceedings}, volume 1412 of {\em Lecture Notes in Computer
  Science}, pages 153--168. Springer, 1998.

\bibitem[KS10]{DBLP:journals/disopt/KobayashiS10}
Yusuke Kobayashi and Christian Sommer.
\newblock On shortest disjoint paths in planar graphs.
\newblock {\em Discret. Optim.}, 7(4):234--245, 2010.

\bibitem[LMS90]{DBLP:journals/dam/LiMS90}
Chung{-}Lun Li, S.~Thomas McCormick, and David Simchi{-}Levi.
\newblock The complexity of finding two disjoint paths with min-max objective
  function.
\newblock {\em Discret. Appl. Math.}, 26(1):105--115, 1990.

\bibitem[Loc21]{SDP-undirected-original}
William Lochet.
\newblock A polynomial time algorithm for the \emph{k}-disjoint shortest paths
  problem.
\newblock In D{\'{a}}niel Marx, editor, {\em Proceedings of the 2021 {ACM-SIAM}
  Symposium on Discrete Algorithms, {SODA} 2021, Virtual Conference, January 10
  - 13, 2021}, pages 169--178. {SIAM}, 2021.

\bibitem[LS22]{DBLP:journals/tcs/LopesS22}
Raul Lopes and Ignasi Sau.
\newblock A relaxation of the directed disjoint paths problem: {A} global
  congestion metric helps.
\newblock {\em Theor. Comput. Sci.}, 898:75--91, 2022.

\bibitem[RS95]{DBLP:journals/jct/RobertsonS95b}
Neil Robertson and Paul~D. Seymour.
\newblock Graph minors .xiii. the disjoint paths problem.
\newblock {\em J. Comb. Theory, Ser. {B}}, 63(1):65--110, 1995.

\bibitem[RT87]{DBLP:journals/combinatorica/RaghavanT87}
Prabhakar Raghavan and Clark~D. Thompson.
\newblock Randomized rounding: a technique for provably good algorithms and
  algorithmic proofs.
\newblock {\em Comb.}, 7(4):365--374, 1987.

\bibitem[RTZ08]{DBLP:journals/talg/RodittyTZ08}
Liam Roditty, Mikkel Thorup, and Uri Zwick.
\newblock Roundtrip spanners and roundtrip routing in directed graphs.
\newblock {\em {ACM} Trans. Algorithms}, 4(3):29:1--29:17, 2008.

\bibitem[RW19]{DBLP:journals/sigact/RubinsteinW19}
Aviad Rubinstein and Virginia~Vassilevska Williams.
\newblock {SETH} vs approximation.
\newblock {\em {SIGACT} News}, 50(4):57--76, 2019.

\bibitem[Sch94a]{schefflerpractical}
Petra Scheffler.
\newblock {\em A practical linear time algorithm for disjoint paths in graphs
  with bounded tree-width}.
\newblock TU, Fachbereich 3, 1994.

\bibitem[Sch94b]{DBLP:journals/siamcomp/Schrijver94}
Alexander Schrijver.
\newblock Finding k disjoint paths in a directed planar graph.
\newblock {\em {SIAM} J. Comput.}, 23(4):780--788, 1994.

\bibitem[SLK{\etalchar{+}}90]{korte1990paths}
Alexander Schrijver, Laszlo Lovasz, Bernhard Korte, Hans~Jurgen Promel, and
  R.~L. Graham.
\newblock {\em Paths, Flows, and VLSI-Layout}.
\newblock Springer-Verlag, Berlin, Heidelberg, 1990.

\bibitem[Wu10]{DBLP:conf/cocoa/Wu10}
Bang~Ye Wu.
\newblock A simpler and more efficient algorithm for the next-to-shortest path
  problem.
\newblock In Weili Wu and Ovidiu Daescu, editors, {\em Combinatorial
  Optimization and Applications - 4th International Conference, {COCOA} 2010,
  Kailua-Kona, HI, USA, December 18-20, 2010, Proceedings, Part {II}}, volume
  6509 of {\em Lecture Notes in Computer Science}, pages 219--227. Springer,
  2010.

\end{thebibliography}

\appendix

\section{Related Work on Disjoint Paths Problems}\label{app:rel}

The problem of finding disjoint paths in a graph is a classical problem, which has been studied extensively (e.g., for applications in VLSI layout \cite{korte1990paths}). 
In the standard formulation of the problem, we are given node pairs $(s_1,t_1),\dots,(s_k,t_k)$ and the goal is to find $k$ node-disjoint paths between each $s_i$ and its corresponding $t_i$. 
When $k$ is part of the input, this problem is known to be \textsf{NP}-complete \cite{DBLP:journals/siamcomp/EvenIS76}. 
For fixed $k$ in undirected graphs, Robertson and Seymour \cite{DBLP:journals/jct/RobertsonS95b} showed that this problem can be solved in  $O(n^3)$ time.
This running time has since been improved to $O(n^2)$
\cite{DBLP:journals/jct/KawarabayashiKR12}.

For directed graphs, this problem is much harder. 
Even for $k=2$, the problem is  \textsf{NP}-complete \cite{DBLP:journals/tcs/FortuneHW80}. 
However, it can be solved in polynomial time for fixed $k$ on certain special classes of directed graphs, such as DAGs \cite{DBLP:journals/tcs/FortuneHW80}, planar graphs, \cite{DBLP:journals/siamcomp/Schrijver94,DBLP:conf/focs/CyganMPP13}, tournament graphs \cite{chudnovsky2015disjoint}, and graphs of bounded treewidth \cite{schefflerpractical}.
These results on restricted graph classes, besides being interesting in their own right, have also found applications in showing that certain ``detour'' problems are fixed-parameter tractable on those same restricted graph classes \cite{detours}.

It is natural to try to \emph{minimize the lengths} of the disjoint paths. There are several ways to define this, such as the \emph{min-sum} version where the goal is to minimize the sum of the lengths of the disjoint paths, or the \emph{min-max} version where the goal is to minimize the maximum length of a path in the solution. Because these variants are generalizations of the original problem, they remain \textsf{NP}-complete for directed graphs even for $k=2$. However, even for undirected graphs, these problems pose barriers. 
The min-max version is \textsf{NP}-hard even for $k=2$ \cite{DBLP:journals/dam/LiMS90}. 
The min-sum version is in polynomial time for $k=2$ \cite{bjorklund2014shortest} and its complexity is open for $k\ge 3$. Considerable effort has been put into solving the min-sum version for $k\ge 3$, however the most general graph classes for which it is known to be in polynomial time are planar graphs with certain additional restrictions on the placement of the terminals relative to the faces \cite{DBLP:journals/disopt/KobayashiS10,DBLP:journals/talg/VerdiereS11,DBLP:conf/esa/BorradaileNZ15,DBLP:conf/fsttcs/DattaIK018}.

With the intractability of the above length minimization variants in mind, it is natural to study the $k$-disjoint shortest paths problem ($k$-DSP) which imposes that \emph{all paths} in the solution must be shortest. See \cref{sec:back} for the known results on $k$-DSP.

\section{Counterexample for Generalizing \texorpdfstring{\cref{lem:dag}}{Lemma 2.3} to Directed Graphs}\label{app:counter}

Recall \cref{lem:dag}:

\lemdag*

We will show a counterexample for extending \cref{lem:dag} to general directed graphs. The counterexample graph is simply a directed cycle on $n$ nodes. Let $k=n$ and let $c=3n/4+1$, so $d=n/4-1$, so indeed the assumption $k>3d$ holds. 
The input pairs $(s_1,t_1),\dots,(s_k,t_k)$ for the $(k,c)$-SPC instance  are defined as follows. Suppose the nodes are numbered from 1 to $n$ around the cycle. Then, for all $i$ from 1 to $n$, $s_i$ is defined as node $i$ in the cycle, and $t_i$ is node $i+3n/4 \pmod n$. Notice that there is only one path from each $s_i$ to the corresponding $t_i$. Taking these paths together, for all $i$, forms a valid $(k,c)$-SPC solution where every node has congestion $c=3n/4+1$; that is, every node is a max-congestion node. This is the only solution because there is a unique shortest path between each pair $(s_i,t_i)$. That is, there does not exist a solution where some solution path $P_i$ passes through all max-congestion nodes. This completes the counterexample. The parameters can be adjusted to allow for a constant other than 3 in the assumption $k>3d$.

\section{\texorpdfstring{\cref{lm:heavy-path}}{Lemma 2.5} implies \texorpdfstring{\cref{thm:undir}}{Lemma 1.4}}\label{app:cor}

We recall the statements of \cref{lm:heavy-path,thm:undir} below:

\lmheavy*

\thmundir*

The proof that \cref{lm:heavy-path} implies \cref{thm:undir} is almost the same as the analogous proof for DAGs in \cite{SDP-congestion-DAG}. 

\begin{proof}[Proof of \Cref{thm:undir} given \cref{lm:heavy-path}]

The following lemma is a generalization of \cite[Lemma 5]{SDP-congestion-DAG} which applies to undirected graphs instead of DAGs. 
Its proof is identical to the proof in \cite{SDP-congestion-DAG}, except that one repeatedly applies \Cref{lm:heavy-path} instead of applying an analogous claim for DAGs, and $(3d,2d)$ is replaced with $(4d,3d)$.

\begin{lemma}
	\label{lm:structure}
	Let $k > 4d$.
	Then any instance of $(k,c)$-SPC either has no solution, or has a solution consisting of paths $P_1, \dots, P_k$ such that there exists a subset 
	$I\sub \set{1, \dots, k}$ of size $4d$ so that the paths $\set{P_i}_{i\in I}$ solve the $(4d, 3d)$-SPC
	problem with sources $\set{s_i}_{i\in I}$ and targets $\set{t_i}_{i\in I}$.
	
\end{lemma}

We also need the following lemma from \cite{SDP-congestion-DAG} which observes that $(k,c)$-SPC can be reduced to $k$-DSP 
on a slightly larger graph.

\begin{lemma}[{\cite[Corollary 2]{SDP-congestion-DAG}}]
    \label{eq:reduce-uncongested}
    If $k$-DSP can be solved in $f(n,k)$ time, then $(k,c)$-SPC can be solved in $f(cn, k)$ time.
\end{lemma}

Now we will use these lemmas to prove the result.

    If $k \le 4d$, then $c\le 4d$ as well.
    Thus, by \Cref{eq:reduce-uncongested}, we can solve the given $(k,c)$-SPC instance in $f(n,4d)$
    time, which certainly implies the desired result.
    
    Otherwise, $k > 4d$.
    In this case, by \Cref{lm:structure} we know that if the $(k,c)$-SPC instance is solvable, it must have a solution where some subset of $4d$ (source, target) pairs solve a $(4d, 3d)$-SPC instance. 
    If we find a set of $4d$ paths solving such an instance, then we can easily complete it to a solution to the original $(k,c)$-SPC problem simply by finding
    any shortest paths between the remaining source and target nodes.
    This works because the congestion of the resulting $k$ paths is at most 
    $3d + (k-4d) = k-d = c.$
    
        Thus, we can solve the $(k,c)$-SPC instance by trying out all $\binom{k}{4d}$ possible subsets of $4d$ (source, target) pairs, attempting to solve $(4d,3d)$-SPC with respect to these terminal nodes, and if we find a solution, conclude by finding $k-4d$ shortest paths between the remaining terminal nodes (e.g., by Dijkstra's algorithm). 
    By \Cref{eq:reduce-uncongested}, each of the $(4d,3d)$-SPC instances constructed in this way can be solved in $f(3dn,4d)$ time. 
    It follows that we can solve the original problem in 
        $O(\binom{k}{4d}f(3dn, 4d))$
    time as desired.
\end{proof}

\end{document}